%% file: draft.tex
\pdfoutput=1

\RequirePackage[l2tabu,orthodox]{nag}

\documentclass[11pt,letterpaper,reqno,oneside]{article}

\input{preamble.tex}

\usepackage{sidecap}
\sidecaptionvpos{figure}{c}

\usepackage{centernot}

\makeatletter
\def\namedlabel#1#2{\begingroup
	#2%
	\def\@currentlabel{#2}%
	\phantomsection\label{#1}\endgroup
}
\makeatother

\newcommand{\SC}{\mathrel{S}}
	\newcommand{\nSC}{\mathrel{\centernot{S}}}

\newcommand{\pref}{\succ}
	\newcommand{\prefeq}{\succeq}
	\newcommand{\npref}{\nsucc}
	\newcommand{\nprefeq}{\nsucceq}
\newcommand{\primgeq}{\mathrel{\unrhd}}
	\newcommand{\primg}{\mathrel{\rhd}}
	\newcommand{\primleq}{\mathrel{\unlhd}}
	
	\newcommand{\nprimgeq}{\mathrel{\centernot{\unrhd}}}
\newcommand{\abstrgeq}{\sqsupseteq}
	\newcommand{\abstrg}{\sqsupset}
	\newcommand{\abstrleq}{\sqsubseteq}
	\newcommand{\abstrl}{\sqsubset}

\def\comp{\mathrel{%
	\mathchoice{\COMP}{\COMP}{\scriptstyle\COMP}{\scriptscriptstyle\COMP}%
}}
\def\COMP{{%
	\setbox0\hbox{$\abstrgeq\hspace{2.3pt}$}%
	\rlap{\hbox to \wd0{\hss$\hspace{2.3pt}\abstrleq$\hss}}\box0
}}
\def\compsim{\mathrel{%
	\mathchoice{\COMPSIM}{\COMPSIM}{\scriptstyle\COMPSIM}{\scriptscriptstyle\COMPSIM}%
}}
\def\COMPSIM{{%
	\setbox0\hbox{$\primgeq\hspace{1.5pt}$}%
	\rlap{\hbox to \wd0{\hss$\hspace{1.5pt}\primleq$\hss}}\box0
}}

\title{\scshape Single-crossing dominance:\\A preference lattice%
\thanks{We are grateful to
Eddie Dekel,
Péter Es\H{o},
Alessandro Pavan,
John Quah and
Bruno Strulovici
for guidance and comments,
and to
Pierpaolo Battigalli,
Paweł Dziewulski,
Matteo Escudé,
Peter Klibanoff,
Massimo Marinacci,
Meg Meyer,
Efe Ok,
Daniel Rappoport,
Kevin Reffett,
Todd Sarver,
Eddie Schlee,
Eran Shmaya,
Marciano Siniscalchi,
Lorenzo Stanca,
Tomasz Strzalecki,
audiences at Northwestern,
and the editor, an associate editor and two anonymous referees
for helpful comments and suggestions.
Curello acknowledges support from the German Research Foundation (DFG) through CRC TR 224 (Project B02).}}

\author{Gregorio Curello \\ University of Mannheim
\and Ludvig Sinander \\ University of Oxford}

\date{6 October 2025}

\makeatletter
	\AtBeginDocument{ \hypersetup{
		pdftitle = {Single-crossing dominance: A preference lattice},
		pdfauthor = {Gregorio Curello and Ludvig Sinander}
		} }
\makeatother

\begin{document}

\maketitle

\begin{abstract}
	Most comparisons of preferences are instances of \emph{single-crossing dominance.}
	We examine the lattice structure of single-crossing dominance, proving characterisation, existence and uniqueness results for minimum upper bounds of arbitrary sets of preferences.
	We apply these theorems
	to derive new comparative statics theorems for collective choice and under analyst uncertainty, and
	to characterise a general `maxmin' class of uncertainty-averse preferences over Savage acts.
\end{abstract}

\section{Introduction}
\label{sec:introduction}

Comparisons of preferences are ubiquitous in economics: examples include
`more risk-averse/uncertainty-averse than' (in decision theory),%
	\footnote{\textcite{Yaari1969,Epstein1999,GhirardatoMarinacci2002}.}
`takes larger actions than' (in monotone comparative statics),%
	\footnote{\textcite{Topkis1978,MilgromShannon1994,QuahStrulovici2009}.}
`more aligned than' (in agency and social-choice contexts),%
	\footnote{\textcite{Kemeny1959,Grandmont1978,PuppeSlinko2019,CurelloSinander2023,CurelloSinander2024,KartikKleiner2024}.}
`more impatient than' (in dynamic problems),%
	\footnote{\textcite{Horowitz1992,BenoitOk2007}.}
and `more self-controlled than' (in models of temptation).%
	\footnote{\textcite{GulPesendorfer2001,DekelLipman2012}.}
All of these preference comparisons, and many others besides, are special cases of \emph{single-crossing dominance,} a general unified way of comparing preferences.

In this paper, we investigate the lattice structure of single-crossing dominance.
Our results characterise the minimum upper bounds of arbitrary sets of preferences, and furnish necessary and sufficient conditions for their existence and uniqueness.

We use these theorems to derive new insights in two economic settings that feature some of the aforementioned preference comparisons.
First, we derive comparative-statics theorems for collective choice
and for problems in which the analyst has only partial knowledge of a decision-maker's preference.
Second, we characterise a general class of maxmin preferences over Savage acts as minimum upper bounds with respect to `more uncertainty-averse than'.

\subsection{Overview of the theory}
\label{sec:introduction:theory}

The abstract environment consists of a non-empty set $\mathcal{X}$ of alternatives equipped with a partial order $\primgeq$.%
	\footnote{Standard order-theoretic concepts and notation are reviewed in \cref{app:appendix:definitions}.}
A \emph{preference} is a complete and transitive binary relation on $\mathcal{X}$.
We write $\mathcal{P}$ for the set of all preferences.
\emph{Single-crossing dominance}
captures a greater appetite for $\primgeq$-larger alternatives:
\begin{definition}
	\label{definition:SC}
	For two preferences $\mathord{\prefeq},\mathord{\prefeq'} \in \mathcal{P}$,
	we write $\mathord{\prefeq'} \SC \mathord{\prefeq}$ iff for any pair $x \primgeq y$ in $\mathcal{X}$, $x \prefeq \mathrel{(\pref)} y$ implies $x \prefeq' \mathrel{(\pref')} y$.
	We say that $\prefeq'$ \emph{single-crossing dominates} $\prefeq$.%
		\footnote{This definition is from \textcite{MilgromShannon1994}, in the context of monotone comparative statics.
		This type of preference comparison has been used at least since Yaari's (\citeyear{Yaari1969}) definition of `more risk-averse than' (see §\ref{sec:maxmin:risk}) and Kemeny's (\citeyear{Kemeny1959}) `betweenness'. The name `single-crossing' comes from the fact that given a set $\Theta \subseteq \R$ and, for each $\theta \in \Theta$, a preference $\mathord{\prefeq_\theta} \in \mathcal{P}$ and a utility representation $u(\cdot,\theta) : \mathcal{X} \to \R$ of $\prefeq_\theta$, the family $(\mathord{\prefeq}_\theta)_{\theta \in \Theta}$ is $\SC$-increasing (i.e. $\mathord{\prefeq}_{\theta'} \SC \mathord{\prefeq}_\theta$ for any pair $\theta' > \theta$ in $\Theta$) iff for any pair $x \primgeq y$ in $\mathcal{X}$, the function $\theta \mapsto u(x,\theta) - u(y,\theta)$ crosses zero at most once, and if so then from below: that is, $u(x,\theta) - u(y,\theta) \geq \mathrel{(>)} 0$ implies $u(x,\theta') - u(y,\theta') \geq \mathrel{(>)} 0$ for $\theta'>\theta$ in $\Theta$.}
\end{definition}

$\SC$ is a binary relation on $\mathcal{P}$. It is reflexive and transitive, but not necessarily anti-symmetric.%
	\footnote{Consider $\mathcal{X} = \{x,y\}$ with the empty partial order $\primgeq$, i.e. $x \nprimgeq y \nprimgeq x$, and let $\mathord{\prefeq},\mathord{\prefeq'} \in \mathcal{P}$ satisfy $x \pref y$ and $y \pref' x$. Then $\mathord{\prefeq} \SC \mathord{\prefeq'} \SC \mathord{\prefeq}$.}

Given a set $P \subseteq \mathcal{P}$ of preferences, a preference $\mathord{\prefeq'} \in \mathcal{P}$ is an \emph{upper bound} of $P$ iff $\mathord{\prefeq'} \SC \mathord{\prefeq}$ for every $\mathord{\prefeq} \in P$,
and a \emph{minimum} upper bound iff in addition $\mathord{\prefeq''} \SC \mathord{\prefeq'}$ for every (other) upper bound $\mathord{\prefeq''}$ of $P$.
(Minimum upper bounds are also known as `least upper bounds', `suprema' or `joins'.)
Intuitively, a minimum upper bound of $P$ is a preference that likes large alternatives more than does any preference in $P$, but only just.
Maximum lower bounds (also known as `greatest lower bounds', `infima' or `meets') are defined analogously.

In §\ref{sec:theory}, we study the \emph{lattice structure} of $(\mathcal{P},\mathord{\SC})$ by developing characterisation, existence and uniqueness results for minimum upper bounds.
Our \hyperref[theorem:characterisation]{characterisation theorem} (§\ref{sec:theory:characterisation}) describes the minimum upper bounds of arbitrary sets $P \subseteq \mathcal{P}$ of preferences.
Our \hyperref[theorem:existence]{existence theorem} (§\ref{sec:theory:existence}) identifies the condition on $\primgeq$, called \emph{crown- and diamond-freeness,} that is necessary and sufficient for every set $P \subseteq \mathcal{P}$ to possess a minimum upper bound.
Finally, our \hyperref[proposition:uniqueness]{uniqueness proposition} (§\ref{sec:theory:uniqueness}) asserts that every set $P \subseteq \mathcal{P}$ has a \emph{unique} minimum upper bound precisely if $\primgeq$ is \emph{complete.}
We extend our results to maximum lower bounds in \cref{app:appendix:meets}.

\subsection{Overview of the applications}
\label{sec:introduction:applications}

We employ our theorems to two economic domains.

\begin{namedthm}[Application to monotone comparative statics {\normalfont (§\ref{sec:compstat})}.]
	In monotone comparative statics, an agent chooses an alternative from a set $\mathcal{X} \subseteq \R$.
	The canonical result states that if the agent's preference $\mathord{\prefeq} \in \mathcal{P}$ increases in the sense of single-crossing dominance $\SC$,
	then her optimal choices $X(\mathord{\prefeq}) = \left\{ x \in \mathcal{X} : \text{$x \prefeq y$ for every $y \in \mathcal{X}$} \right\}$
	increase in the the strong set order.

	We first study collective choice:
	there is a group of agents,
	with preferences $P \subseteq \mathcal{P}$.
	We prove that when the set $P$ increases in the strong set order, so does the \emph{consensus:} the set $C(P) = \Intersect_{ \mathord{\prefeq} \in P } X(\mathord{\prefeq})$ of alternatives that every agent considers optimal.
	Generalising, we characterise comparative statics for the set $C_k(P)$ of alternatives that every individual considers at least $k^\text{th}$-best:
	it increases in the strong set order whenever $P$ does also if $k=2$ or (trivially) if $k = \abs*{\mathcal{X}}$,
	but can strictly \emph{decrease} if $2 < k < \abs*{\mathcal{X}}$.
	The proofs of these results make extensive use of the \hyperref[theorem:existence]{existence} and \hyperref[theorem:characterisation]{characterisation} theorems.%
		\footnote{Furthermore, the very possibility of comparing sets $P,P' \subseteq \mathcal{P}$ by the strong set order relies on the \hyperref[proposition:uniqueness]{uniqueness proposition}, since the strong set order is defined in terms of minimum upper bounds and maximum lower bounds.}

	Secondly, we consider comparative-statics predictions by an analyst who knows only that the agent's preference belongs to a set $P \subseteq \mathcal{P}$.
	The possibly-optimal choices from a menu $M \subseteq \mathcal{X}$
	are $X_M(P) \coloneqq \Union_{\mathord{\prefeq} \in P} X_M(\mathord{\prefeq})$, where $X_M(\mathord{\prefeq})$ are the optimal choices for $\mathord{\prefeq} \in \mathcal{P}$.
	Under a richness assumption on $P$,
	possible choices are sharply bounded by
	the choices of the minimum upper bound $\prefeq^\star$ of $P$:
	$\max X_M(P) = \max X_M(\mathord{\prefeq^\star})$ for any menu $M \subseteq \mathcal{X}$.
	The proof turns on the \hyperref[theorem:characterisation]{characterisation theorem}, and the existence of $\mathord{\prefeq^\star}$ is ensured by the \hyperref[theorem:existence]{existence theorem}.
	Comparative statics follow:
	for a shift of the uncertainty $P$ to increase $\min X_M(P)$ and $\max X_M(P)$ whatever the menu $M \subseteq \mathcal{X}$,
	it is necessary and sufficient that $P$'s minimum upper bound and maximum lower bound both increase.
\end{namedthm}

\begin{namedthm}[Application to uncertainty- and risk-aversion {\normalfont (§\ref{sec:maxmin})}.]
	In the Savage framework, there are (potentially payoff-relevant) \emph{consequences} and possible \emph{states of the world,}
	and a decision-maker has preferences over \emph{acts,} meaning maps from states to consequences.
	Let $\mathcal{X}$ be the set of all acts, and $\mathcal{P}$ the set of all preferences over acts.

	One preference is called \emph{more uncertainty-averse than} another iff whenever the latter (strictly) prefers an unambiguous act to some other act, so does the former.
	`Unambiguous act' can mean a constant act,
	or more generally one that is measurable with respect to an (exogenously-given) collection of events deemed unambiguous.
	For the rest of this summary, assume that only constant acts are unambiguous, and that consequences are monetary prizes.

	We consider preferences that are \emph{monotone} (more money for sure is better than less money for sure)
	and \emph{solvable} (every act has a certainty equivalent).
	For a monotone and solvable preference $\mathord{\prefeq} \in \mathcal{P}$,
	write $e(\mathord{\prefeq},x) \in \R$ for the (by monotonicity, unique) certainty equivalent of an act $x \in \mathcal{X}$.

	A \emph{maxmin preference} is one that is represented by $x \mapsto \min_{\mathord{\prefeq} \in P} e(\mathord{\prefeq},x)$ for some set $P \subseteq \mathcal{P}$ of monotone and solvable preferences.
	Intuitively, such a preference cautiously values each act $x \in \mathcal{X}$ according to its worst certainty equivalent among the preferences in $P$.
	Maxmin expected utility \parencite{GilboaSchmeidler1989} is the special case in which $P$ comprises only expected-utility preferences with common risk attitude (but different beliefs).

	We characterise maxmin preferences by proving that $P$ is a maxmin representation of $\prefeq^\star$ iff $\prefeq^\star$ is a minimum upper bound of $P$ with respect to `more uncertainty-averse than'.
	The key insight underlying the proof is that if $\mathord{\prefeq},\mathord{\prefeq'} \in \mathcal{P}$ are monotone, then $\prefeq'$ is more uncertainty-averse than $\prefeq$ exactly if $\mathord{\prefeq'} \SC \mathord{\prefeq}$, where $\SC$ is the single-crossing dominance relation induced by the partial order $\primgeq$ on $\mathcal{X}$ defined as follows: $x \primgeq y$ iff either (i)~$x=y$, (ii)~$x$ is constant and $y$ is not, or (iii)~$x,y$ are both constant and $x>y$.

	This result carries over to risk-aversion:
	\emph{cautious} preferences over lotteries are precisely minimum upper bounds with respect to `more risk-averse than'. The key here is that minimum upper bounds with respect to `more risk-averse than' are exactly minimum upper bounds with respect to $\SC$ induced by the partial order $\primgeq$ defined by (i)--(iii) above.
\end{namedthm}

\subsection{Related literature}
\label{sec:introduction:lit}

Our work relates to the combinatorics literature on permutation lattices \parencite[e.g.][]{BennettBirkhoff1994,Markowsky1994,DuquenneCherfouh1994}.
Here the alternatives are $\mathcal{X} = \{1,\dots,n\}$ for some $n \in \N$,
and the partial order $\primgeq$ is the ordinary inequality.
Thus $\primgeq$ is complete and $\mathcal{X}$ is finite.

In this context, anti-symmetric (i.e. never-indifferent) preferences may be thought of as \emph{permutations,}
and single-crossing dominance is known as the \emph{weak order} (or \emph{permutohedron order}).
It has been known since \textcite{GuilbaudRosenstiehl1963} and \textcite{YanagimotoOkamoto1969} that the set of all permutations equipped with the weak order is a (complete) lattice.
Our \hyperref[proposition:uniqueness]{uniqueness proposition} is a result along these lines.

Since this literature assumes that $\primgeq$ is complete, it certainly contains no analogue of our \hyperref[theorem:existence]{existence theorem}.
We are not aware of any analogue of our \hyperref[theorem:characterisation]{characterisation theorem}, either.
Besides avoiding the restrictive assumption that $\primgeq$ is complete, we differ from this literature by allowing for preferences with indifferences and by permitting $\mathcal{X}$ to be of unrestricted cardinality.

Crown- and diamond-freeness are standard concepts in combinatorics \parencite[e.g.][]{BakerFishburnRoberts1972,GriggsLiLu2012,Lu2014}.
\textcite{BallPultrSichler2006} show that crown- and diamond-freeness together are \emph{nearly} equivalent to the absence of weak cycles from the Hasse diagram.
This latter property appears in the probability literature,
where it characterises those posets $(\mathcal{X},\mathord{\primgeq})$ on which every first-order stochastically increasing family of probability measures
may be realised by an a.s. increasing process on $\mathcal{X}$ \parencite{FillMachida2001}.
\textcite{BrooksFrankelKamenica2022} prove the analogous result for second-order stochastic dominance, and apply this to beliefs and information structures.

\section{Theory}
\label{sec:theory}

In this section, we develop our general results about the lattice structure of single-crossing: our \hyperref[theorem:characterisation]{characterisation theorem} (§\ref{sec:theory:characterisation}), our \hyperref[theorem:existence]{existence theorem} (§\ref{sec:theory:existence}) and our \hyperref[proposition:uniqueness]{uniqueness proposition} (§\ref{sec:theory:uniqueness}).
Recall from §\ref{sec:introduction:theory} the abstract environment and basic definitions.

\subsection{Characterisation of minimum upper bounds}
\label{sec:theory:characterisation}

Our characterisation will be in terms of \emph{$P$-chains,} defined as follows.

\begin{definition}
	\label{definition:Pchains}
	For a set $P \subseteq \mathcal{P}$ of preferences and two alternatives $x \primgeq y$ in $\mathcal{X}$, a \emph{$P$-chain} from $x$ to $y$ is a finite sequence $(w_k)_{k=1}^K$ in $\mathcal{X}$ such that
	\begin{enumerate}[label=(\roman*)]

		\item $w_1 = x$ and $w_K = y$,

		\item for every $k<K$, $w_k \primgeq w_{k+1}$, and

		\item for every $k<K$, $w_k \prefeq w_{k+1}$ for some $\mathord{\prefeq} \in P$.

	\end{enumerate}
	A \emph{strict $P$-chain} is a $P$-chain with $w_k \pref w_{k+1}$ for some $k<K$ and $\mathord{\prefeq} \in P$.
\end{definition}

In words, a $P$-chain is a $\primgeq$-decreasing sequence of alternatives along which, at each juncture (each $k<K$), some preference in $P$ prefers the previous ($\primgeq$-larger) alternative to the subsequent ($\primgeq$-smaller) one.
Clearly a $P$-chain of length $K \geq 3$ is simply the concatenation of $K-1$ $P$-chains of length 2.

\setcounter{example}{0}
\begin{example}
	\label{example:hook_Pchain}
	Consider $\mathcal{X} = \{x,y,z,w\}$, with $\primgeq$ such that
	$x \primg w$ and $x \primg y \primg z$ (so $x \primg z$), and $w,y$ and $w,z$ are $\primg$-incomparable.
	The partial order $\primgeq$ may be depicted graphically as
	\begin{equation*}
		\vcenter
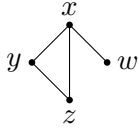
{\hbox{%
		\begin{tikzpicture}[scale=0.5]

			\fill ( 0,2) circle[radius=2.5pt];
			\fill (-1,1) circle[radius=2.5pt];
			\fill ( 0,0) circle[radius=2.5pt];
			\fill ( 1,1) circle[radius=2.5pt];

			\draw[-] ( 0,0)--(-1,1)--( 0,2)--( 1,1);

			\draw ( 0,2) node[anchor=south] {$x$};
			\draw (-1,1) node[anchor=east ] {$y$};
			\draw ( 0,0) node[anchor=north] {$z$};
			\draw ( 1,1) node[anchor=west ] {$w$};

		\end{tikzpicture}%
		}}
	\end{equation*}
	In this (`Hasse') diagram, there is path from $a$ \emph{down} to $b$ iff $a \primgeq b$. We will use diagrams of this sort throughout.

	Consider $P = \{ \mathord{\prefeq_1}, \mathord{\prefeq_2} \}$,
	$z \pref_1 w \pref_1 x \pref_1 y$
	and
	$y \pref_2 z \pref_2 w \pref_2 x$.
	By inspection, $(x,y)$ and $(y,z)$ are strict $P$-chains.
	Thus there is a strict $P$-chain from $x$ to $z$, namely $(x,y,z)$.
	Note, however, that $(x,z)$ is \emph{not} a $P$-chain, since neither preference favours $x$ over $z$.

	Although $x \primgeq w$, there is no $P$-chain from $x$ to $w$: the only candidate is $(x,w)$, and it fails to be a $P$-chain since neither preference favours $x$ over $w$.
\end{example}

The following asserts that an upper bound of $P \subseteq \mathcal{P}$ is precisely a preference that (strictly) prefers a larger alternative to a smaller one whenever there is a (strict) $P$-chain between them:

\begin{lemma}
	\label{lemma:UB_characterisation}
	For a preference $\mathord{\prefeq'} \in \mathcal{P}$ and a set $P \subseteq \mathcal{P}$ of preferences, the following are equivalent:
	\begin{enumerate}

		\item \label{bullet:UB1}
		$\prefeq'$ is an upper bound of $P$.

		\item \label{bullet:UB2}
		$\prefeq'$ satisfies: for any $x,y \in \mathcal{X}$ such that $x \primgeq y$,
		\begin{enumerate}[label=(\roman*)]

			\item \label{bullet:UB_weak}
			$x \prefeq' y$ if there is a $P$-chain from $x$ to $y$, and

			\item \label{bullet:UB_strict}
			$x \pref' y$ if there is a strict $P$-chain from $x$ to $y$.

		\end{enumerate}

	\end{enumerate}
\end{lemma}

\setcounter{example}{0}
\begin{example}[continued]
	\label{example:hook_UB}
	The $P$-chains, all of them strict, are $(x,y)$, $(y,z)$ and $(x,y,z)$.
	Thus by \Cref{lemma:UB_characterisation}, a preference $\mathord{\prefeq'} \in \mathcal{P}$ is an upper bound of $P$ iff $x \pref' y \pref' z$ (and $x \pref' z$).
	Thus
	$x \pref'_a y \pref'_a z \pref'_a w$
	and
	$w \pref'_b x \pref'_b y \pref'_b z$
	are both upper bounds.
\end{example}

\begin{proof}
	\emph{\ref{bullet:UB2} implies \ref{bullet:UB1}:}
	Let $\mathord{\prefeq'} \in \mathcal{P}$ satisfy condition \ref{bullet:UB2}; we wish to show that $\mathord{\prefeq'} \SC \mathord{\prefeq}$ for any $\mathord{\prefeq} \in P$. To that end, fix a pair $x,y \in \mathcal{X}$ such that $x \primgeq y$, and suppose that $x \prefeq \mathrel{(\pref)} y$ for some $\mathord{\prefeq} \in P$; we must show that $x \prefeq' \mathrel{(\pref')} y$. This is immediate since $(x,y)$ is a (strict) $P$-chain.

	\emph{\ref{bullet:UB1} implies \ref{bullet:UB2}:}
	Let $\prefeq'$ be an upper bound of $P$.
	Fix a pair $x,y \in \mathcal{X}$ such that $x \primgeq y$; we must show that if there is a (strict) $P$-chain from $x$ to $y$, then $x \prefeq' \mathrel{(\pref')} y$.

	Suppose that there exists a $P$-chain $( w_k )_{k=1}^K$ from $x$ to $y$. For each $k < K$, we have $w_k \primgeq w_{k+1}$ as well as $w_k \prefeq_k w_{k+1}$ for some $\mathord{\prefeq_k} \in P$.
	Because $\prefeq'$ is an upper bound of $P$, it must be that $w_k \prefeq' w_{k+1}$ for each $k < K$.
	Since $\prefeq'$ is transitive (because it lives in $\mathcal{P}$), it follows that $x \prefeq' y$.

	Suppose there is a strict $P$-chain $( w_k )_{k=1}^K$ from $x$ to $y$. As in the weak case, we must have $w_k \prefeq' w_{k+1}$ for each $k < K$. Moreover, since the $P$-chain is strict, we have $w_k \pref w_{k+1}$ for some $k < K$ and $\mathord{\prefeq} \in P$; hence $w_k \pref' w_{k+1}$ since $\prefeq'$ is an upper bound of $P$.
	Thus $x \pref' y$ by the transitivity of $\prefeq'$.
\end{proof}

\Cref{lemma:UB_characterisation} says that an upper bound must have a (strict) `upward' preference whenever there is a (strict) $P$-chain. Our \hyperref[theorem:characterisation]{characterisation theorem} says that the \emph{minimum} upper bounds are those which have a (strict) `upward' preference \emph{only} when there is a (strict) $P$-chain:

\begin{namedthm}[Characterisation theorem.]
	\label{theorem:characterisation}
	For a preference $\mathord{\prefeq^\star} \in \mathcal{P}$ and a set $P \subseteq \mathcal{P}$ of preferences, the following are equivalent:
	\begin{enumerate}

		\item \label{bullet:charac1}
		$\prefeq^\star$ is a minimum upper bound of $P$.

		\item \label{bullet:charac2}
		$\prefeq^\star$ satisfies: for any $x,y \in \mathcal{X}$ such that $x \primgeq y$,
		\begin{enumerate}[label=(2\alph*)]

			\item \label{bullet:join_weak}
			$x \prefeq^\star y$ iff there is a $P$-chain from $x$ to $y$, and

			\item \label{bullet:join_strict}
			$x \pref^\star y$ iff there is a strict $P$-chain from $x$ to $y$.

		\end{enumerate}

	\end{enumerate}
\end{namedthm}

The analogous result for maximum lower bounds is given in \cref{app:appendix:meets}.

\setcounter{example}{0}
\begin{example}[continued]
	\label{example:hook_MUB}
	A minimum upper bound ${\prefeq^\star} \in \mathcal{P}$ must satisfy $x \pref^\star y \pref^\star z$ (and $x \pref^\star z$) since it is an upper bound.
	Since $x \primgeq w$ but there is no $P$-chain from $x$ to $w$, minimumhood requires that $w \pref^\star x$ by the \hyperref[theorem:characterisation]{characterisation theorem}.
	In sum, ${\prefeq^\star} \in \mathcal{P}$ is a minimum upper bound iff
	$w \pref^\star x \pref^\star y \pref^\star z$.
\end{example}

One direction of the proof is straightforward:

\begin{proof}[Proof that \ref{bullet:charac2} implies \ref{bullet:charac1}]
	\label{proof:thm1_21_join}
	Fix a subset $P$ of $\mathcal{P}$ and a $\mathord{\prefeq^\star} \in \mathcal{P}$ that satisfies \ref{bullet:join_weak}--\ref{bullet:join_strict}.
	It is immediate from \Cref{lemma:UB_characterisation} that $\prefeq^\star$ is an upper bound of $P$.
	To see that $\prefeq^\star$ is a minimum of the upper bounds of $P$, let $\prefeq'$ be any upper bound of $P$. Fix a pair $x,y \in \mathcal{X}$ such that $x \primgeq y$, and suppose that $x \prefeq^\star \mathrel{(\pref^\star)} y$. By property \ref{bullet:join_weak} (property \ref{bullet:join_strict}), there must be a (strict) $P$-chain from $x$ to $y$. Since $\prefeq'$ is an upper bound of $P$, it follows by \Cref{lemma:UB_characterisation} that $x \prefeq' \mathrel{(\pref')} y$. Since $x,y \in \mathcal{X}$ were arbitrary, this establishes that $\mathord{\prefeq'} \SC \mathord{\prefeq^\star}$.
\end{proof}

The other direction relies on the following lemma, whose proof is given in \cref{app:appendix:lemma_other_UB_pf}.

\begin{lemma}
	\label{lemma:other_UB}
	Let $P$ be a set of preferences, and let $x,y \in \mathcal{X}$ satisfy $x \primgeq y$. If there is no (strict) $P$-chain from $x$ to $y$, then there exists an upper bound $\prefeq''$ of $P$ with $x \nprefeq'' \mathrel{(\npref'')} y$.
\end{lemma}

\begin{proof}[Proof that \ref{bullet:charac1} implies \ref{bullet:charac2}]
	\label{proof:thm1_join_21}
	Fix a set $P \subseteq \mathcal{P}$ of preferences and a preference $\mathord{\prefeq'} \in \mathcal{P}$.
	We will prove the contra-positive: if $\prefeq'$ violates \ref{bullet:join_weak}--\ref{bullet:join_strict}, then it cannot be a minimum upper bound of $P$.
	If a preference $\prefeq'$ violates the `if' part of \ref{bullet:join_weak} or the `only if' part of \ref{bullet:join_strict},
	then it fails to be an upper bound of $P$ by \Cref{lemma:UB_characterisation}.

	Suppose that a preference $\prefeq'$ violates the `only if' part of \ref{bullet:join_weak} (the `if' part of \ref{bullet:join_strict}):
	there are $x,y \in \mathcal{X}$ with $x \primgeq y$ such that there is no (strict) $P$-chain from $x$ to $y$, and yet $x \prefeq' \mathrel{(\pref')} y$. By \Cref{lemma:other_UB}, there is an upper bound $\prefeq''$ of $P$ such that $x \nprefeq'' \mathrel{(\npref'')} y$. Then $\mathord{\prefeq''} \nSC \mathord{\prefeq'}$, so $\prefeq'$ fails to be a minimum of the upper bounds of $P$.
\end{proof}

\subsection{Existence of minimum upper bounds}
\label{sec:theory:existence}

In this section, we provide a necessary and sufficient condition on $\primgeq$ for minimum upper bounds to exist for every set of preferences.
This condition rules out two special subposets: \emph{crowns} and \emph{diamonds.}

In §\ref{sec:theory:existence:crowns}, we define crowns and show that $\primgeq$ must be free of them if every set of preferences is to possess a minimum upper bound.
In §\ref{sec:theory:existence:diamonds}, we do the same for diamonds.
In §\ref{sec:theory:existence:theorem}, we give the \hyperref[theorem:existence]{existence theorem}, which asserts in addition that crown- and diamond-freeness of $\primgeq$ is \emph{sufficient} for existence.

\subsubsection{Crowns}
\label{sec:theory:existence:crowns}

The following example shows how existence can fail.

\begin{namedthm}[Crown example.]
	\label{example:crown_no_join}
	Consider $\mathcal{X} = \{x,y,z,w\}$ with the following order $\primgeq$:
	\begin{equation*}
		\vcenter{\hbox{%
		\begin{tikzpicture}[scale=0.5]

			\fill (0,2) circle[radius=2.5pt];
			\fill (0,0) circle[radius=2.5pt];
			\fill (2,2) circle[radius=2.5pt];
			\fill (2,0) circle[radius=2.5pt];

			\draw[-] (0,2) -- (0,0) -- (2,2) -- (2,0) -- cycle;

			\draw (0,2) node[anchor=south east] {$x$};
			\draw (0,0) node[anchor=north east] {$y$};
			\draw (2,2) node[anchor=south west] {$z$};
			\draw (2,0) node[anchor=north west] {$w$};

		\end{tikzpicture}
		}}
	\end{equation*}
	(That is, each of $x$ and $z$ $\primgeq$-dominates each of $y$ and $w$, but $x,z$ and are $\primgeq$-incomparable, as are $y,w$.)
	Let $P = \{ \mathord{\prefeq_1}, \mathord{\prefeq_2} \} \subseteq \mathcal{P}$, where
	$w \pref_1 x \pref_1 y \pref_1 z$
	and
	$y \pref_2 z \pref_2 w \pref_2 x$.
	We have $x \primgeq y$ and $z \primgeq w$, and there is a strict $P$-chain from $x$ to $y$ and from $z$ to $w$.
	On the other hand, $x \primgeq w$ and $z \primgeq y$, but there is no $P$-chain from $x$ to $w$ or from $z$ to $y$.
	So by the \hyperref[theorem:characterisation]{characterisation theorem}, a minimum upper bound $\prefeq^\star$ of $P$ must have
	$x \pref^\star y \pref^\star z \pref^\star w \pref^\star x$.
	Such a $\prefeq^\star$ cannot be transitive, so cannot live in $\mathcal{P}$. It follows that no minimum upper bound exists.
	(To illustrate,
	consider the preferences $\mathord{\prefeq'},\mathord{\prefeq''} \in \mathcal{P}$ given by $x \pref' y \pref' z \pref' w$ and $z \pref'' w \pref'' x \pref'' y$.
	Both are upper bounds, but neither is minimum since $\mathord{\prefeq'} \nSC \mathord{\prefeq''} \nSC \mathord{\prefeq'}$.)
\end{namedthm}

The problem is that $\primgeq$ features a \emph{crown,} defined as follows.

\begin{definition}
	\label{definition:crown}
	Let $\abstrgeq$ be a binary relation on a set $\mathcal{A}$.
	For $K \geq 4$ even, a \emph{$K$-crown}
	is a sequence $( a_k )_{k=1}^K$ in $\mathcal{A}$ such that
	non-adjacent $a_k,a_{k'}$ are $\abstrgeq$-incomparable,
	and $a_{k-1} \abstrg a_k \abstrl a_{k+1}$ for each even $k \in \{2,\dots,K\}$ (where $a_{K+1} \coloneqq a_1$ by convention).
	A \emph{crown} is a $K$-crown for some $K \geq 4$ even.
	The relation $\abstrgeq$ is \emph{crown-free} iff it features no crowns.
\end{definition}

Some crowns are drawn in \Cref{fig:crowns}.

\begin{figure}
	\begin{subfigure}{0.32\textwidth}
		\centering
		\begin{tikzpicture}[scale=0.5]

			\fill (0,2) circle[radius=2.5pt];
			\fill (0,0) circle[radius=2.5pt];
			\fill (2,2) circle[radius=2.5pt];
			\fill (2,0) circle[radius=2.5pt];

			\draw[-] (0,2) -- (0,0) -- (2,2) -- (2,0) -- cycle;

			\draw (0,2) node[anchor=south] {$a_1$};
			\draw (0,0) node[anchor=north] {$a_2$};
			\draw (2,2) node[anchor=south] {$a_3$};
			\draw (2,0) node[anchor=north] {$a_4$};

		\end{tikzpicture}
		\caption{A $4$-crown.}
		\label{fig:crowns:4crown}
	\end{subfigure}
	\begin{subfigure}{0.32\textwidth}
		\centering
		\begin{tikzpicture}[scale=0.5]

			\fill (0,2) circle[radius=2.5pt];
			\fill (0,0) circle[radius=2.5pt];
			\fill (2,2) circle[radius=2.5pt];
			\fill (2,0) circle[radius=2.5pt];
			\fill (4,2) circle[radius=2.5pt];
			\fill (4,0) circle[radius=2.5pt];

			\draw[-] (0,2) -- (0,0) -- (2,2) -- (2,0)
				-- (4,2) -- (4,0) -- cycle;

			\draw (0,2) node[anchor=south] {$a_1$};
			\draw (0,0) node[anchor=north] {$a_2$};
			\draw (2,2) node[anchor=south] {$a_3$};
			\draw (2,0) node[anchor=north] {$a_4$};
			\draw (4,2) node[anchor=south] {$a_5$};
			\draw (4,0) node[anchor=north] {$a_6$};

		\end{tikzpicture}
		\caption{A $6$-crown.}
		\label{fig:crowns:6crown}
	\end{subfigure}
	\begin{subfigure}{0.32\textwidth}
		\centering
		\begin{tikzpicture}[scale=0.5]

			\fill (0,2) circle[radius=2.5pt];
			\fill (0,0) circle[radius=2.5pt];
			\fill (2,2) circle[radius=2.5pt];
			\fill (2,0) circle[radius=2.5pt];
			\fill (4,2) circle[radius=2.5pt];
			\fill (4,0) circle[radius=2.5pt];
			\fill (6,2) circle[radius=2.5pt];
			\fill (6,0) circle[radius=2.5pt];

			\draw[-] (0,2) -- (0,0) -- (2,2) -- (2,0)
				-- (4,2) -- (4,0)
				-- (6,2) -- (6,0) -- cycle;

			\draw (0,2) node[anchor=south] {$a_1$};
			\draw (0,0) node[anchor=north] {$a_2$};
			\draw (2,2) node[anchor=south] {$a_3$};
			\draw (2,0) node[anchor=north] {$a_4$};
			\draw (4,2) node[anchor=south] {$a_5$};
			\draw (4,0) node[anchor=north] {$a_6$};
			\draw (6,2) node[anchor=south] {$a_7$};
			\draw (6,0) node[anchor=north] {$a_8$};

		\end{tikzpicture}
		\caption{An $8$-crown.}
		\label{fig:crowns:8crown}
	\end{subfigure}
	\caption{Crowns.}
	\label{fig:crowns}
\end{figure}

Crown-freeness rules out a specific form of incompleteness. A strong sufficient condition is completeness. A weaker sufficient condition is that the comparability relation $\comp$ be transitive.%
	\footnote{For a binary relation $\abstrgeq$, comparability $\comp$ is defined by $a \comp b$ iff either $a \abstrgeq b$ or $b \abstrgeq a$.}
Neither is necessary:
\begin{namedthm}[Diamond example.]
	\label{example:diamond_nocrown}
	Consider $\mathcal{X}$ with the partial order $\primgeq$ given by
	\begin{equation*}
		\vcenter
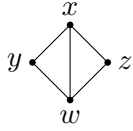
{\hbox{%
		\begin{tikzpicture}[scale=0.5]

			\fill (1,2) circle[radius=2.5pt];
			\fill (0,1) circle[radius=2.5pt];
			\fill (2,1) circle[radius=2.5pt];
			\fill (1,0) circle[radius=2.5pt];

			\draw[-] (1,2) -- (0,1) -- (1,0) -- (2,1) -- (1,2);

			\draw (1,2) node[anchor=south] {$x$};
			\draw (0,1) node[anchor=east ] {$y$};
			\draw (2,1) node[anchor=west ] {$z$};
			\draw (1,0) node[anchor=north] {$w$};

		\end{tikzpicture}
		}}
	\end{equation*}
	(That is, $x \primgeq y \primgeq w$ and $x \primgeq z \primgeq w$, but $y,z$ are $\primgeq$-incomparable.)
	$\primgeq$ is not complete since $y,z$ are $\primgeq$-incomparable.
	Nor is comparability $\compsim$ transitive, since $y \compsim x \compsim z$ but $y,z$ are $\primgeq$-incomparable.
	But $\primgeq$ is manifestly crown-free.
\end{namedthm}

\begin{lemma}[necessity of crown-freeness]
	\label{lemma:crownfree_necessary}
	If every pair of preferences possesses a minimum upper bound, then $\primgeq$ is crown-free.
\end{lemma}

\begin{proof}
	We prove the contra-positive.
	Suppose that $\primgeq$ features a crown $(x_1,\dots,x_K)$.
	Consider $P = \{ \mathord{\prefeq_a}, \mathord{\prefeq_b} \} \subseteq \mathcal{P}$, where
	\begin{align*}
		x_K \pref_a x_1 \pref_a x_2 &\pref_a \cdots
		\pref_a x_{K-2} \pref_a x_{K-1}
		\\
		x_2 &\pref_b x_3 \pref_b \cdots
		\pref_b x_{K-1} \pref_b x_K \pref_b x_1 .
	\end{align*}
	For even $k \in \{2,\dots,K\}$, $x_{k-1} \primg x_k$ and $x_{k-1} \pref x_k$ for some $\mathord{\prefeq} \in P$ (in particular, $\prefeq_a$ for $k<K$ even, $\prefeq_b$ for $k>1$ even). Hence $(x_{k-1},x_k)$ is a strict $P$-chain, so by the \hyperref[theorem:characterisation]{characterisation theorem}, $x_{k-1} \pref^\star x_k$ for any minimum upper bound $\prefeq^\star$ of $P$.

	Moreover, for even $k \in \{2,\dots,K\}$, $x_{k+1} \primg x_k$, and $x_{k+1} \nprefeq x_k$ for all $\mathord{\prefeq} \in P$. (This is apparent, separately, for $k < K$ even and for $k=K$.) Hence there is no $P$-chain from $x_{k+1}$ to $x_k$, so by the \hyperref[theorem:characterisation]{characterisation theorem}, $x_k \pref^\star x_{k+1}$ for any minimum upper bound of $\prefeq^\star$ of $P$.

	It follows that any minimum upper bound $\prefeq^\star$ of $P$ must satisfy
	$x_1 \pref^\star x_2 \pref^\star \cdots
	\pref^\star x_{K-1} \pref^\star x_K \pref^\star x_1$.
	Such a $\prefeq^\star$ cannot be transitive, so cannot live in $\mathcal{P}$; hence $P$ admits no minimum upper bound.
\end{proof}

\subsubsection{Diamonds}
\label{sec:theory:existence:diamonds}

Existence can fail even in the absence of crowns:

\begin{namedthm}[Diamond example \textnormal{(continued)}.]
	\label{example:diamond_no_join}
	Let $P = \{ \mathord{\prefeq_1}, \mathord{\prefeq_2} \} \subseteq \mathcal{P}$, where
	$y \pref_1 w \pref_1 z \pref_1 x$
	and
	$w \pref_2 z \pref_2 x \pref_2 y$.
	Evidently $(x,y)$ and $(y,w)$ are strict $P$-chains.
	We have $z \primgeq w$ but no $P$-chain from $z$ to $w$, and $x \primgeq z$ but no $P$-chain from $x$ to $z$. So by the \hyperref[theorem:characterisation]{characterisation theorem}, a minimum upper bound $\prefeq^\star$ of $P$ must satisfy
	$x \pref^\star y \pref^\star w \pref^\star z \pref^\star x$.
	Since such a relation $\prefeq^\star$ cannot be transitive, it follows that $P$ admits no minimum upper bound.
\end{namedthm}

The trouble is that $\primgeq$ features a \emph{diamond:}

\begin{definition}
	\label{definition:diamond-free}
	Let $\abstrgeq$ be a partial order on a set $\mathcal{A}$.
	A \emph{diamond} is a length-4 sequence $(a,b,c,d)$ in $\mathcal{A}$ such that $a \abstrgeq b \abstrgeq d$ and $a \abstrgeq c \abstrgeq d$, but $b,c$ are $\abstrgeq$-incomparable.
	The order $\abstrgeq$ is \emph{diamond-free} iff it features no diamonds.
\end{definition}

Diamond-freeness is strong; for example, a diamond-free lattice is necessarily totally ordered.%
	\footnote{\label{footnote:diamond_lattice}If $(\mathcal{X},\primgeq)$ is a lattice but not totally ordered, then there are $\primgeq$-incomparable $x,y \in \mathcal{X}$, in which case $( x \meet y, x, y, x \join y)$ is a diamond, where `$x \meet y$' (`$x \join y$') denotes the maximum lower bound (minimum upper bound) of $\{x,y\}$.}
Like crown-freeness, diamond-freeness rules out a specific form of incompleteness,
and is implied by strong forms of `limited incompleteness' such as completeness of $\abstrgeq$ or transitivity of the comparability relation $\comp$.
Neither of these conditions is necessary for diamond-freeness, nor is crown-freeness:

\begin{namedthm}[Crown example \textnormal{(continued)}.]
	\label{example:crown_diamond-free}
	By inspection, there are no diamonds.
	But there is a crown,
	and $\primgeq$ is not complete,
	nor is $\compsim$ is transitive.
\end{namedthm}

\begin{lemma}[necessity of diamond-freeness]
	\label{lemma:diamondfree_necessary}
	If every pair of preferences possesses a minimum upper bound, then $\primgeq$ is diamond-free.
\end{lemma}

The proof is almost exactly the \hyperref[example:diamond_no_join]{diamond example}, so we omit it.

\subsubsection{The \texorpdfstring{\hyperref[theorem:existence]{existence theorem}}{existence theorem}}
\label{sec:theory:existence:theorem}

Our \hyperref[theorem:existence]{existence theorem} asserts that crown- and diamond-freeness are not only necessary,
but also sufficient, for the existence of minimum upper bounds.

\begin{namedthm}[Existence theorem.]
	\label{theorem:existence}
	The following are equivalent:
	\begin{enumerate}

		\item \label{bullet:ex_ub}
		Every \emph{set} of preferences has a minimum upper bound.

		\item \label{bullet:ex_ub2}
		Every \emph{pair} of preferences has a minimum upper bound.

		\item \label{bullet:ex_red}
		$\primgeq$ is crown- and diamond-free.

	\end{enumerate}
\end{namedthm}

By way of illustration, crown- and diamond-freeness fails in
the \hyperref[example:crown_no_join]{crown}
and \hyperref[example:diamond_nocrown]{diamond}
examples,
but is satisfied in \Cref{example:hook_Pchain}.
More generally, it holds whenever there are three or fewer alternatives, and fails for any lattice $(\mathcal{X},\primgeq)$ that is not totally ordered (see \cref{footnote:diamond_lattice}).
Crown- and diamond-freeness is arguably a strong property, but it does hold in some applications (see §\ref{sec:compstat}--\ref{sec:maxmin}).

In \cref{app:appendix:meets}, we show further that these properties are equivalent to every set (or two-element set) of preferences possessing a maximum lower bound.
By analogy with complete lattices, call $(\mathcal{P},\mathord{\SC})$ a \emph{pre-lattice} iff every two-element set $P \subseteq \mathcal{P}$ possesses a minimum upper bound and maximum lower bound,
and a \emph{complete pre-lattice} if this holds for \emph{every} set $P \subseteq \mathcal{P}$.%
	\footnote{$(\mathcal{P},\mathord{\SC})$ is not a complete lattice
	because $\SC$ is a mere pre-order, i.e. it need not be anti-symmetric.
	The failure of anti-symmetry means that minimum upper bounds and maximum lower bounds need not be unique; we study this issue in §\ref{sec:theory:uniqueness}.}	
The \hyperref[corollary:P_pre-lattice_join_meet]{elaborated existence theorem} in \cref{app:appendix:meets} implies that $(\mathcal{P},\mathord{\SC})$ is a complete pre-lattice iff it is a pre-lattice iff $\primgeq$ features no crowns or diamonds.

As for the proof, it is immediate that \ref{bullet:ex_ub} implies \ref{bullet:ex_ub2}, and we have already established in \Cref{lemma:crownfree_necessary,lemma:diamondfree_necessary} that \ref{bullet:ex_ub2} implies \ref{bullet:ex_red}.
Proving that \ref{bullet:ex_red} implies \ref{bullet:ex_ub} is more difficult; we do this in \cref{app:appendix:existence_pf}.
The idea is as follows.
Let $\prefeq^\circ$ be the minimal binary relation (in general incomplete) that satisfies properties \ref{bullet:join_weak}--\ref{bullet:join_strict} in the \hyperref[theorem:characterisation]{characterisation theorem}.
We show first that absent diamonds in $\primgeq$, $\prefeq^\circ$ must be `weakly transitive'.
We then show that when there are no crowns in $\primgeq$, weak transitivity of $\prefeq^\circ$ implies that it satisfies a stronger transitivity-type property called \emph{Suzumura consistency.}
This permits us to invoke an extension theorem due to \textcite{Richter1966,Suzumura1976} to conclude that $\prefeq^\circ$ may be extended to a complete and transitive relation (i.e. a preference).
This preference is a minimum upper bound by the \hyperref[theorem:characterisation]{characterisation theorem}.

\subsection{Uniqueness of minimum upper bounds}
\label{sec:theory:uniqueness}

When minimum upper bounds exist, they need not be unique:

\setcounter{example}{1}
\begin{example}
	\label{example:xy_not_unique}
	Consider $\mathcal{X} = \{x,y\}$ with the empty partial order $\primgeq$, so that $x,y$ are $\primgeq$-incomparable.
	Let $P = \{ \mathord{\prefeq_1}, \mathord{\prefeq_2} \} \subseteq \mathcal{P}$, where $x \pref_1 y$ and $y \pref_2 x$.

	Since all alternatives are $\primgeq$-incomparable, there are no $P$-chains. Conditions \ref{bullet:join_weak}--\ref{bullet:join_strict} in the \hyperref[theorem:characterisation]{characterisation theorem} are therefore (vacuously) satisfied by any preference. So by the \hyperref[theorem:characterisation]{characterisation theorem}, every preference is a minimum upper bound of $P$.
\end{example}

The message of \Cref{example:xy_not_unique} is that preferences that disagree only on $\primgeq$-incomparable pairs of alternatives $\SC$-dominate each other, leading to a multiplicity of minimum upper bounds. This is a general lesson:

\begin{namedthm}[Uniqueness proposition.]
	\label{proposition:uniqueness}
	The following are equivalent:
	\begin{enumerate}

		\item \label{bullet:uniqueness_unique}
		Every set of preferences has \emph{at most} one minimum upper bound.

		\item \label{bullet:uniqueness_exist}
		Every set of preferences has \emph{exactly} one minimum upper bound.

		\item \label{bullet:uniqueness_compl}
		$\primgeq$ is complete.

	\end{enumerate}
\end{namedthm}

The analogue for maximum lower bounds is given in \cref{app:appendix:meets}. Together, these two results imply that $(\mathcal{P},\mathord{\SC})$ is a complete lattice iff $\primgeq$ is complete.

\begin{proof}[Proof that \ref{bullet:uniqueness_compl} implies \ref{bullet:uniqueness_exist}]
	Suppose that $\primgeq$ is complete, and fix a set $P \subseteq \mathcal{P}$.
	Since completeness implies crown- and diamond-freeness, $P$ has at least one minimum upper bound by the \hyperref[theorem:existence]{existence theorem}.

	To show uniqueness, let $\mathord{\prefeq'},\mathord{\prefeq''} \in \mathcal{P}$ be minimum upper bounds of $P \subseteq \mathcal{P}$.
	Then by the \hyperref[theorem:characterisation]{characterisation theorem}, $\mathord{\prefeq'}$ and $\mathord{\prefeq''}$ must agree on all $\primgeq$-comparable pairs of alternatives.
	Since $\primgeq$ is complete, it follows that $\mathord{\prefeq'}$ and $\mathord{\prefeq''}$ agree on all pairs of alternatives, i.e. that they are identical.
\end{proof}

It is immediate that \ref{bullet:uniqueness_exist} implies \ref{bullet:uniqueness_unique}.
Our proof that \ref{bullet:uniqueness_unique} implies \ref{bullet:uniqueness_compl} relies on the following lemma, proved in \cref{app:appendix:2_UBs_pf}.

\begin{lemma}
	\label{lemma:2_UBs}
	Let $x,y \in \mathcal{X}$ be $\primgeq$-incomparable.
	Then any set $P \subseteq \mathcal{P}$ has upper bounds $\mathord{\prefeq'}, \mathord{\prefeq''} \in \mathcal{P}$ such that $x \pref' y$ and $y \pref'' x$.
\end{lemma}

\begin{proof}[Proof that \ref{bullet:uniqueness_unique} implies \ref{bullet:uniqueness_compl}]
	We prove the contra-positive.
	Suppose that $\primgeq$ is incomplete.
	Then the grand set $\mathcal{P}$ has multiple upper bounds by \Cref{lemma:2_UBs},
	and clearly each of these is a minimum upper bound.
\end{proof}

\section{Application to monotone comparative statics}
\label{sec:compstat}

In this section, we use our theorems to extend the theory of monotone comparative statics%
	\footnote{See e.g. \textcite{Topkis1978,MilgromShannon1994,QuahStrulovici2009},
	and the textbook treatment by \textcite{Topkis1998}.}
in two directions.

In §\ref{sec:compstat:consensus},
we study the comparative statics of collective choice.
We focus on the \emph{consensus:} 
the set of alternatives that every individual in a group $P \subseteq \mathcal{P}$ considers optimal.
We show that when $P$ increases in the strong set order, so does the consensus.
We go on fully to characterise the comparative-statics properties of the set of of alternatives that each individual considers at least $k^\text{th}$-best:
this set increases in the strong set order as $P$ does also
if $k=2$ or (trivially) if $k=\abs*{\mathcal{X}}$,
but can strictly \emph{decrease} if $2 < k < \abs*{\mathcal{X}}$.

In §\ref{sec:compstat:robust},
we consider the problem of an analyst who wishes to predict how an agent will choose from any given menu $M \subseteq \mathcal{X}$ of alternatives, but is uncertain about the agent's preference: all she knows is that it belongs to a set $P \subseteq \mathcal{P}$.
We show that the minimum upper bound and maximum lower bound of $P$ provide tight bounds on possible choice, and sharply characterise how possible choices vary with the analyst's uncertainty $P$.

\subsection{The canonical theory}
\label{sec:compstat:canonical}

We begin with a brief recap of the canonical theory of monotone comparative statics.
An agent chooses an alternative $x$ from a set $\mathcal{X} \subseteq \R$ ordered by the usual inequality $\geq$.
The agent chooses optimally with respect to her preference $\mathord{\prefeq} \in \mathcal{P}$.
Denote by $X(\mathord{\prefeq})$ the (possibly empty) set of optimal alternatives for preference $\mathord{\prefeq} \in \mathcal{P}$:
\begin{equation*}
	X(\mathord{\prefeq}) \coloneqq
	\left\{
	x \in \mathcal{X} : \text{$x \prefeq y$ for every $y \in \mathcal{X}$}
	\right\} .
\end{equation*}

\begin{definition}
	\label{definition:SSO}
	Consider a lattice $(\mathcal{A},\mathord{\abstrgeq})$ and two subsets $A,B \subseteq \mathcal{A}$. $A$ dominates $B$ in the \emph{($\abstrgeq$-induced) strong set order} iff for any $a \in A$ and $b \in B$, the minimum upper bound of $\{a,b\}$ belongs to $A$ and the maximum lower bound of $\{a,b\}$ belongs to $B$.
\end{definition}

Note that $A$ dominates $B$ in the strong set order if either set is empty. The canonical monotone comparative statics (`MCS') result is the following.

\begin{namedthm}[MCS theorem.\footnote{Due to \textcite{MilgromShannon1994,LicalziVeinott1992}.}]
	\label{theorem:MS}
	Let $\mathord{\prefeq},\mathord{\prefeq'} \in \mathcal{P}$ be preferences.
	If $\mathord{\prefeq'} \SC \mathord{\prefeq}$, then $X(\mathord{\prefeq'})$ dominates $X(\mathord{\prefeq})$ in the ($\geq$-induced) strong set order.
\end{namedthm}

That is,
when the agent's preference increases in the sense of $\SC$, the set of optimal alternatives increases in the sense of the strong set order.

\begin{remark}
	\label{remark:diamond-free_lattice}
	The most general \hyperref[theorem:MS]{MCS theorem} allows the set of actions to be any lattice $(\mathcal{X},\mathord{\primgeq})$.
	This added generality is not useful for our purposes because in order to apply the \hyperref[theorem:existence]{existence theorem}, we shall require that $\primgeq$ be crown- and diamond-free, and any diamond-free lattice is necessarily totally ordered (see \cref{footnote:diamond_lattice}).
\end{remark}

\subsection{Comparative statics for collective choice}
\label{sec:compstat:consensus}

There is a group of agents, each with a preference $\mathord{\prefeq} \in \mathcal{P}$.
Write $P \subseteq \mathcal{P}$ for the set of preferences in the group.
The \emph{consensus} $C(P)$ is the set of alternatives that \emph{every} individual in the group finds optimal:
$C(P)
\coloneqq \Intersect_{ \mathord{\prefeq} \in P }
X(\mathord{\prefeq})$.

Since $\geq$ is complete, $(\mathcal{P},\SC)$ is a lattice by the \hyperref[proposition:uniqueness]{uniqueness proposition}. We may therefore use the ($\SC$-induced) strong set order to compare sets of preferences.

\begin{proposition}[consensus comparative statics]
	\label{proposition:compstat_consensus}
	Let $\mathcal{X}$ be a subset of $\R$ ordered by inequality $\geq$,
	and let $P,P' \subseteq \mathcal{P}$ be sets of preferences.
	If $P'$ dominates $P$ in the ($\SC$-induced) strong set order, then $C(P')$ dominates $C(P)$ in the ($\geq$-induced) strong set order.
\end{proposition}

In other words, when agents' preferences shift up in the sense of the strong set order, so does the consensus.
Note that it may be that either $C(P)$ or $C(P')$ is empty, in which case the conclusion holds automatically.

\begin{proof}
	Fix $P,P' \subseteq \mathcal{P}$ such that $P'$ dominates $P$ in the $\SC$-induced strong set order.
	The conclusion is immediate if either $C(P)$ or $C(P')$ is empty, so suppose not.
	Take $x \in C(P)$ and $x' \in C(P')$, and let $x \meet x'$ ($x \join x'$) denote the maximum lower bound (minimum upper bound) of $\{x,x'\}$; we must show that $x \meet x'$ lies in $C(P)$ and that $x \join x'$ lies in $C(P')$.
	We will prove the former; the proof of the latter is similar.

	Take any $\mathord{\prefeq} \in P$ and $\mathord{\prefeq'} \in P'$.
	Since the order $\geq$ on $\mathcal{X}$ is complete, the set $\{ \mathord{\prefeq}, \mathord{\prefeq'} \}$ possesses a minimum upper bound $\prefeq^\star$ by the \hyperref[theorem:existence]{existence theorem}.
	Since $P'$ dominates $P$ in the $\SC$-induced strong set order, the minimum upper bound $\prefeq^\star$ lies in $P'$.
	Because $\mathord{\prefeq^\star} \SC \mathord{\prefeq}$, the \hyperref[theorem:MS]{MCS theorem} implies that $X(\mathord{\prefeq^\star})$ dominates $X(\mathord{\prefeq})$ in the $\geq$-induced strong set order.
	Since $x \in C(P) \subseteq X(\mathord{\prefeq})$ and $x' \in C(P') \subseteq X(\mathord{\prefeq^\star})$, it follows that $x \meet x' \in X(\mathord{\prefeq})$.
	Since $\mathord{\prefeq} \in P$ was arbitrary, this shows that $x \meet x' \in C(P)$.
\end{proof}

\Cref{proposition:compstat_consensus} can be used to study comparative statics for social choice.
A \emph{preference profile} is an element $\pi = (\mathord{\prefeq_1},\dots,\mathord{\prefeq_k})$ of $\Union_{n \in \N} \mathcal{P}^n$, and its \emph{support} is the set $\supp \pi \coloneqq \{ \mathord{\prefeq_1},\dots,\mathord{\prefeq_k} \}$ of all preferences represented in it. A \emph{social choice function (SCF)} is a map $\phi : \Union_{n \in \N} \mathcal{P}^n \to \mathcal{X}$ that picks an alternative for each preference profile. It is \emph{monotone} iff $\phi(\pi') \geq \phi(\pi)$ whenever $\supp \pi'$ dominates $\supp \pi$ in the ($\SC$-induced) strong set order,
and \emph{respects unanimity} iff $\phi(\pi) \in C(\supp \pi)$ whenever the latter is non-empty.
The following corollary of \Cref{proposition:compstat_consensus} is proved in \cref{app:supp_appendix:consensus_selection}:

\begin{corollary}
	\label{corollary:consensus_selection}
	Let $\mathcal{X} \subseteq \R$ be compact.
	Then there exists a monotone SCF that respects unanimity.
\end{corollary}

When the consensus is empty,
it is natural to consider individuals' \emph{second-}favourite alternatives.
Let the \emph{second consensus} be
$C_2(P)
\coloneqq \Intersect_{\mathord{\prefeq} \in P}
X_2(\mathord{\prefeq})$,
where $X_2(\mathord{\prefeq})$ are
the alternatives that $\mathord{\prefeq} \in \mathcal{P}$
considers at least second-best:
\begin{equation*}
	X_2(\mathord{\prefeq})
	= \left\{
	x \in \mathcal{X} :
	\text{$y \pref x$
	for at most one $y \in \mathcal{X}$}
	\right\} .
\end{equation*}

\begin{proposition}[second consensus comparative statics]
	\label{proposition:compstat_consensus_2}
	Let $\mathcal{X}$ be a subset of $\R$ ordered by inequality $\geq$,
	and let $P,P' \subseteq \mathcal{P}$ be sets of preferences.
	If $P'$ dominates $P$ in the ($\SC$-induced) strong set order
	and $C(P)=C(P')$,
	then $C_2(P') \setminus C(P')$ dominates $C_2(P) \setminus C(P)$ in the ($\geq$-induced) strong set order.
\end{proposition}

In particular, the second consensus increases ($C_2(P')$ dominates $C_2(P)$ in the strong set order)
whenever the consensus is empty ($C(P) = \varnothing = C(P')$).
Outside of that case, however, the second consensus need not increase.%
	\footnote{If $\mathcal{X} = \{1,2,3\}$,
	$P = \{\mathord{\prefeq}\}$ where $1 \pref 2 \pref 3$
	and $P' = \{\mathord{\prefeq},\mathord{\prefeq'}\}$
	where $1 \pref' 3 \pref' 2$,
	then $C(P) = \{1\} = C(P')$, $C_2(P) = \{1,2\}$ and $C_2(P') = \{1\}$.}

The proof is in \cref{app:supp_appendix:compstat_consensus_2}.
The argument is rather intricate, and makes extensive use of minimum upper bounds and maximum lower bounds.

Surprisingly, for $2 < k < \abs*{\mathcal{X}}$, there is no analogous comparative-statics result for the set $C_k(P)$ of alternatives that are at least $k^\text{th}$-best according to every $\mathord{\prefeq} \in P$.
The following counter-example is for $k=3$ and $\abs*{\mathcal{X}}=4$,
but extends easily to any $k$ and $\abs*{\mathcal{X}}$ such that $2 < k < \abs*{\mathcal{X}}$.

\setcounter{example}{2}
\begin{example}
	\label{example:top3}
	Let $\mathcal{X} = \{1,2,3,4\}$.
	For any labelling $\{x,y,z,w\} = \mathcal{X}$,
	write `$xyzw$' for the preference $\mathord{\prefeq} \in \mathcal{P}$ satisfying $x \pref y \pref z \pref w$.
	Let
	\begin{align*}
		P &\coloneqq \{ 1234, 1324, 1342, 2134, 2314, 2341, 3214, 3241, 3421 \} \quad \text{and}
		\\
		P' &\coloneqq \{ 4321, 4231, 4213, 3421, 3241, 2341, 3214, 2314, 2134 \} .
	\end{align*}
	Then $P'$ dominates $P$ in the ($\SC$-induced) strong set order,%
		\footnote{Moreover, $P$ is a \emph{sublattice} (it dominates itself in the strong set order), and so is $P'$.}
	and $C(P) = C(P') = C_2(P) = C_2(P') = \varnothing$,
	but $C_3(P) = \{3\}$ and $C_3(P') = \{2\}$.
\end{example}

\begin{remark}
	\label{remark:compstat_consensus_menu}
	The results of this section remain true if agents are constrained to choose from a (possibly small) menu $M \subseteq \mathcal{X}$, provided the menu is \emph{order-convex} in the sense that $x,z \in M$, $y \in \mathcal{X}$ and $x \leq y \leq z$ imply $y \in M$.
	For example, \Cref{proposition:compstat_consensus} generalises as follows: under the same hypothesis, for any non-empty and order-convex $M \subseteq \mathcal{X}$, $C_M(P')$ dominates $C_M(P)$ in the ($\geq$-induced) strong set order, where $C_M(P) \coloneqq \Intersect_{ \mathord{\prefeq} \in P }
	X_M(\mathord{\prefeq})$ and
	$X_M(\mathord{\prefeq}) \coloneqq 
 	\left\{
 	x \in M : \text{$x \prefeq y$ for every $y \in M$}
 	\right\}$.
	To see why, write $\prefeq_M$ for the restriction of a preference $\mathord{\prefeq} \in \mathcal{P}$ to $M$,%
		\footnote{That is, $\prefeq_M$ is the binary relation on $M$ such that for any $x,y \in M$, $x \prefeq_M y$ iff $x \prefeq y$.}
	and for any $P \subseteq \mathcal{P}$ let $P_M \coloneqq \{ \mathord{\abstrgeq} : \text{$\mathord{\abstrgeq} = \mathord{\prefeq_M}$ for some $\mathord{\prefeq} \in P$} \}$. Since $M$ is order-convex, the hypothesis of \Cref{proposition:compstat_consensus} implies that $P'_M$ dominates $P_M$ in the strong set order. Thus applying \Cref{proposition:compstat_consensus} (with $M$ in place of $\mathcal{X}$) yields that
	\begin{align*}
		&\{ x \in M : \text{$x \abstrgeq y$ for all $y \in M$ and $\mathord{\abstrgeq} \in P'_M$} \}
		\quad \text{dominates}
		\\
		&\{ x \in M : \text{$x \abstrgeq y$ for all $y \in M$ and $\mathord{\abstrgeq} \in P_M$} \}
		\quad \text{in the strong set order.}
	\end{align*}
	By inspection, the former set equals $C_M(P')$, and the latter is $C_M(P)$.
\end{remark}

\subsection{Robust comparative statics}
\label{sec:compstat:robust}

Consider an analyst who knows only that the agent's preference belongs to a set $P \subseteq \mathcal{P}$,
and wishes to predict choice across menus $M \subseteq \mathcal{X}$.
For a non-empty menu $M \subseteq \mathcal{X}$ of alternatives,
the agent's possible choices are
$X_M(P)
\coloneqq \Union_{ \mathord{\prefeq} \in P }
X_M(\mathord{\prefeq})$,
where
\begin{equation*}
	X_M(\mathord{\prefeq}) \coloneqq 
	\left\{
	x \in M : \text{$x \prefeq y$ for every $y \in M$}
	\right\} 
\end{equation*}
are the $\prefeq$-best alternatives in the menu $M$.
The analyst seeks to bound the agent's choices across menus $M$,
and to predict how these choices vary with the uncertainty $P$.
We focus on \emph{rich} uncertainty, formalised as follows.

\begin{definition}
	\label{definition:rich}
	A set $P \subseteq \mathcal{P}$ is \emph{rich}
	iff for any $x_0,x_1,\dots,x_K \in \mathcal{X}$,
	if there are preferences $\mathord{\prefeq_1},\dots,\mathord{\prefeq_K} \in P$
	such that $x_0 \prefeq_1 x_1 \prefeq_2 x_2 \prefeq_3 \cdots \prefeq_K x_K$,
	then there is a preference $\mathord{\prefeq} \in P$
	such that $x_0 \prefeq x_k$ for every $k \in \{1,\dots,K\}$.
\end{definition}

Since $\geq$ is complete,
the \hyperref[theorem:existence]{existence theorem}
guarantees that the (arbitrary) set $P \subseteq \mathcal{P}$ of preferences
has a minimum upper bound and a maximum lower bound.
By the \hyperref[proposition:uniqueness]{uniqueness proposition}, they are unique.

\begin{proposition}
	\label{proposition:bottomtop_weak}
	Let $\mathcal{X}$ be a finite subset of $\R$ ordered by inequality $\geq$,
	and let $P \subseteq \mathcal{P}$ be non-empty and rich.
	Then $\max X_M(P) = \max X_M(\mathord{\prefeq^\star})$ and
	$\min X_M(P) = \min X_M(\mathord{\prefeq_\star})$
	for every non-empty menu $M \subseteq \mathcal{X}$,
	where $\prefeq^\star$ is the minimum upper bound of $P$
	and $\prefeq_\star$ is the maximum lower bound.
\end{proposition}

In other words,
the choices made by $\prefeq^\star$ and $\prefeq_\star$
are tight bounds on the set of possible choices $X_M(P)$,
across all non-empty menus $M \subseteq \mathcal{X}$.

\begin{proof}
	We prove the claim about maxima, omitting the symmetric argument for minima.
	Fix a non-empty menu $M \subseteq \mathcal{X}$.
	We first show that $\max X_M(P) \leq \max X_M(\mathord{\prefeq^\star})$.
	For every $\mathord{\prefeq} \in P$,
	we have $\mathord{\prefeq^\star} \SC \mathord{\prefeq}$ since $\mathord{\prefeq^\star}$ is an upper bound of $P$,
	so that $X_M(\mathord{\prefeq^\star})$ dominates $X_M(\mathord{\prefeq})$ in the $\geq$-induced strong set order by the \hyperref[theorem:MS]{MCS theorem}, implying that
	$\max X_M(\mathord{\prefeq^\star}) \geq \max X_M(\mathord{\prefeq})$.
	Hence $\max X_M(\mathord{\prefeq^\star}) \geq \max \Union_{\mathord{\prefeq} \in P} X_M(\mathord{\prefeq}) = \max X_M(P)$.

	For the reverse inequality, it suffices to exhibit a preference $\mathord{\prefeq} \in P$ such that $\max X_M(\mathord{\prefeq^\star}) \in X_M(\mathord{\prefeq})$.
	To that end, enumerate the elements of the menu as $M = \{x_0,\dots,x_N\}$
	where for each $n<N$,
	we have either $x_n \pref^\star x_{n+1}$
	or $x_n \prefeq^\star x_{n+1} \prefeq^\star x_n$ and $x_n > x_{n+1}$.
	Evidently $x_0 = \max X_M(\mathord{\prefeq^\star})$.

	\begin{namedthm}[Claim.]
		\label{claim:n_n1_rich}
		For each $n<N$,
		$P$ contains a preference $\prefeq_n$
		such that $x_n \prefeq_n x_{n+1}$.
	\end{namedthm}

	\begin{proof}[Proof of the \protect{\hyperref[claim:n_n1_rich]{claim}}]%
		\renewcommand{\qedsymbol}{$\square$}
		Fix an $n<N$.
		Suppose first that $x_{n+1} > x_n$.
		Then by the \hyperref[theorem:characterisation]{characterisation theorem}, there is no strict $P$-chain from $x_{n+1}$ to $x_n$;
		in particular, $(x_{n+1},x_n)$ is not a strict $P$-chain.
		Thus we may choose $\prefeq_n$ to be any preference in $P$, since all of them prefer $x_{n+1}$ to $x_n$.

		Suppose instead that $x_n > x_{n+1}$.
		Then by the \hyperref[theorem:characterisation]{characterisation theorem}, there must be a $P$-chain $(y^k)_{k=1}^K$ from $x_n$ to $x_{n+1}$,
		so $x_n = y^0 \prefeq^1 y^1 \prefeq^2 \cdots \prefeq^K y^K = x_{n+1}$ for some preferences $\mathord{\prefeq^1},\dots,\mathord{\prefeq^K} \in P$.
		Since $P$ is rich,
		it follows that there is a $\mathord{\prefeq_n} \in P$
		such that $x_n \prefeq_n x_{n+1}$.
	\end{proof}%
	\renewcommand{\qedsymbol}{$\blacksquare$}

	Since $n<N$ was arbitrary, we have shown that there are preferences $\mathord{\prefeq_1},\dots,\mathord{\prefeq_N} \in P$
	such that $x_0 \prefeq_1 x_1 \prefeq_2 x_2 \prefeq_3 \cdots \prefeq_N x_N$.
	Thus since $P$ is rich,
	it contains a preference $\prefeq$
	such that $x_0 \prefeq x_n$
	for every $n \in \{1,\dots,N\}$,
	which is to say that $x_0 \in X_M(\mathord{\prefeq})$.
\end{proof}

\Cref{proposition:bottomtop_weak} delivers a characterisation of comparative statics:

\begin{corollary}[robust comparative statics]
	\label{corollary:mcs_robust}
	Let $\mathcal{X}$ be a finite subset of $\R$ ordered by inequality $\geq$,
	and let $P,P' \subseteq \mathcal{P}$ be non-empty and rich.
	Write $\prefeq^\star$ ($\prefeq^{\star\prime}$) for the minimum upper bound of $P$ (of $P'$), and $\prefeq_\star$ ($\prefeq_\star'$) for the maximum lower bound.
	The following are equivalent:
	\begin{enumerate}

		\item \label{bullet:mcs_robust:mcs}
		$\mathord{\prefeq^{\star\prime}} \SC \mathord{\prefeq^\star}$ and $\mathord{\prefeq_\star'} \SC \mathord{\prefeq_\star}$.
	
		\item \label{bullet:mcs_robust:shift}
		$\max X_M(P') \geq \max X_M(P)$ and $\min X_M(P') \geq \min X_M(P)$ for every non-empty menu $M \subseteq \mathcal{X}$.
	
	\end{enumerate}
\end{corollary}

In other words, the shifts of the analyst's uncertainty $P$ which lead her to predict higher choices from every menu
are exactly those that increase the minimum upper bound and maximum lower bound.

\begin{proof}
	If $\mathord{\prefeq^{\star\prime}} \SC \mathord{\prefeq^\star}$
	and $\mathord{\prefeq_\star'} \SC \mathord{\prefeq_\star}$,
	then for any non-empty $M \subseteq \mathcal{X}$,
	$X_M(\mathord{\prefeq^{\star\prime}})$ dominates $X_M(\mathord{\prefeq^\star})$ and $X_M(\mathord{\prefeq_\star'})$ dominates $X_M(\mathord{\prefeq_\star})$ in the $\geq$-induced strong set order by the \hyperref[theorem:MS]{MCS theorem},
	whence
	\begin{equation*}
		\max X_M(P')
		= \max X_M(\mathord{\prefeq^{\star\prime}})
		\geq \max X_M(\mathord{\prefeq^\star})
		= \max X_M(P)
	\end{equation*}
	by \Cref{proposition:bottomtop_weak},
	and similarly $\min X_M(P') \geq \min X_M(P)$.
	
	For the converse, suppose that $\mathord{\prefeq^{\star\prime}} \nSC \mathord{\prefeq^\star}$ or $\mathord{\prefeq_\star'} \nSC \mathord{\prefeq_\star}$; without loss of generality, the former.
	Then by definition of $\SC$, there is a binary menu $M = \{x,y\} \subseteq \mathcal{X}$ with $x>y$ such that
	either $X_M(\mathord{\prefeq^{\star\prime}}) \not\ni x \in X_M(\mathord{\prefeq^\star})$
	or $X_M(\mathord{\prefeq^{\star\prime}}) \ni y \not\in X_M(\mathord{\prefeq^\star})$.
	So by \Cref{proposition:bottomtop_weak}, we have either
	\begin{equation*}
		\max X_M(P')
		= \max X_M(\mathord{\prefeq^{\star\prime}})
		= y
		< x
		= \max X_M(\mathord{\prefeq^\star})
		= \max X_M(P)
	\end{equation*}
	or (similarly) $\min X_M(P') < \min X_M(P)$.
\end{proof}

Our analysis has focussed on rich preference uncertainty $P \subseteq \mathcal{P}$.
The following refinement of \Cref{proposition:bottomtop_weak} shows that richness is indispensable:
without it, choice cannot be bounded across menus, whether by $\prefeq^\star$ and $\prefeq_\star$ or by any other pair of preferences.

\begin{namedthm}[\Cref*{proposition:bottomtop_weak}$\boldsymbol{^\star}$\textbf{.}]
	\label{proposition:bottomtop}
	For a finite set $\mathcal{X}$ of alternatives
	and a non-empty set $P \subseteq \mathcal{P}$ of preferences over $\mathcal{X}$, the following are equivalent:
	\begin{enumerate}
		
		\item \label{item:rich}
		$P$ is rich.
		
		\item \label{item:consistent}
		For any total order $\primgeq$ on $\mathcal{X}$,
		there is a preference $\mathord{\prefeq} \in \mathcal{P}$ such that $\max X_M(P) = \max X_M(\mathord{\prefeq})$ for every non-empty $M \subseteq \mathcal{X}$.
		
		\item \label{item:tight_bound}
		For any total order $\primgeq$ on $\mathcal{X}$,
		$\max X_M(P) = \max X_M(\mathord{\prefeq^\star})$ for every non-empty $M \subseteq \mathcal{X}$,
		where $\prefeq^\star$ is the minimum upper bound of $P$.
		
	\end{enumerate}
\end{namedthm}

The proof is in \cref{app:supp_appendix:proposition:bottomtop}.
\ref{item:rich}--\ref{item:tight_bound} are also equivalent to the analogues of \ref{item:consistent} and \ref{item:tight_bound} for $\min X_M(\cdot)$ and the maximum lower bound $\prefeq_\star$ of $P$.

\section{Application to uncertainty- and risk-aversion}
\label{sec:maxmin}

In this section, we apply our results to uncertainty-aversion.%
	\footnote{See e.g. \textcite{Ellsberg1961,Schmeidler1989,GilboaSchmeidler1989,KlibanoffMarinacciMukerji2005,MaccheroniMarinacciRustichini2006}.}
We identify a natural compactness condition
under which `maxmin' preferences are precisely minimum upper bounds with respect to `more uncertainty-averse than',
and without which the equivalence generally fails.

In §\ref{sec:maxmin:def}, we introduce general \emph{maxmin preferences,} which nest maxmin expected utility \parencite{GilboaSchmeidler1989}.
We then (§\ref{sec:maxmin:charac}) use our theorems to characterise maxmin preferences as minimum upper bounds, and deduce uncertainty-aversion comparative statics for maxmin preferences.
Finally, in §\ref{sec:maxmin:risk}, we translate these insights to risk-aversion in choice among lotteries.

\subsection{Environment}
\label{sec:maxmin:environment}

There is a state space $(\mathcal{S},\mathcal{E})$ comprising a non-empty set $\mathcal{S}$ of \emph{states of the world} and a $\sigma$-algebra $\mathcal{E}$ of subsets of $\mathcal{S}$, whose members are called \emph{events.}
There is also a non-empty set $\mathcal{C}$ of (potentially payoff-relevant) \emph{consequences.}
A \emph{(Savage) act} is a finite-ranged $\mathcal{E}$-measurable map $\mathcal{S} \to \mathcal{C}$.

\begin{remark}
	\label{remark:AA}
	This is the \textcite{Savage1954} framework. It would be harmless (but unnecessary) to restrict attention to the Anscombe--Aumann (\citeyear{AnscombeAumann1963}) special case, in which $\mathcal{C}$ is the set of all finite-support lotteries over a set $\Pi$ of prizes.
\end{remark}

Following \textcite{Epstein1999}, we suppose that there is an exogenously-given collection $\mathcal{E}^\circ \subseteq \mathcal{E}$ of \emph{unambiguous} events. The unambiguous events are those to which a decision-maker is able to assign probabilities; they are `understood'.
The collection $\mathcal{E}^\circ$ is assumed to contain the universal event $\mathcal{S}$ and to be closed under complementation and countable disjoint union.%
	\footnote{$\mathcal{S} \setminus E \in \mathcal{E}^\circ$ for any $E \in \mathcal{E}^\circ$,
	and $\Union_{n=1}^\infty E_n \subseteq \mathcal{E}^\circ$ for any pairwise disjoint $(E_n)_{n=1}^\infty \subseteq \mathcal{E}^\circ$.}
An \emph{unambiguous act} is one that is $\mathcal{E}^\circ$-measurable;
all other acts are called \emph{ambiguous.}

Let $\mathcal{X}$ be the set of all acts,
with typical elements $x,y \in \mathcal{X}$.
Write $\mathcal{X}^\circ \subseteq \mathcal{X}$ for those that are unambiguous.
Let $\mathcal{P}$ be the set of all preferences over $\mathcal{X}$.

\begin{definition}
	\label{definition:A}
	For two preferences $\mathord{\prefeq},\mathord{\prefeq'} \in \mathcal{P}$, we say that $\prefeq'$ is \emph{more uncertainty-averse than} $\prefeq$
	iff for any unambiguous act $x^\circ \in \mathcal{X}^\circ$ and any act $x \in \mathcal{X}$,
	$x^\circ \prefeq \mathrel{(\pref)} x$ implies $x^\circ \prefeq' \mathrel{(\pref')} x$.
\end{definition}

This definition is from \textcite{Epstein1999}.
\textcite{GhirardatoMarinacci2002} studied the case in which only trivial events are unambiguous ($\mathcal{E}^\circ = \{\mathcal{S},\varnothing\}$), so that the unambiguous acts are precisely the constant ones.

Given the interpretation of the `unambiguous' events $\mathcal{E}^\circ$ as those to which a decision-maker is able to assign probabilities,
it only makes sense to consider preferences that are consistent with a probabilistic belief about $\mathcal{E}^\circ$.
That is, we must restrict attention to the set $\mathcal{P}^\circ$ of preferences $\mathord{\prefeq} \in \mathcal{P}$
that are \emph{probabilistically sophisticated} on $\mathcal{X}^\circ$,
meaning that their restriction $\mathord{\prefeq}|_{\mathcal{X}^\circ}$ to the unambiguous acts $\mathcal{X}^\circ$
is ordinally represented by $x \mapsto U\left( \mu \circ x^{-1} \right)$
for some map $U : \Delta(\mathcal{C}) \to \R$ and some probability measure $\mu : \sigma\left( \mathcal{E}^\circ \right) \to [0,1]$,%
	\footnote{A relation $\sqsupseteq$ on a set $\mathcal{A}$ is \emph{ordinally represented} by $f : \mathcal{A} \to \R$ iff for any $a,b \in \mathcal{A}$, $a \prefeq b$ iff $f(a) \geq f(b)$.}
where $\sigma\left( \mathcal{E}^\circ \right)$ denotes the $\sigma$-algebra generated by $\mathcal{E}^\circ$, $\Delta(\mathcal{C})$ denotes the set of finite-support lotteries over $\mathcal{C}$,
and $\mu \circ x^{-1} \in \Delta(\mathcal{C})$ is the pushforward lottery: $\left(\mu \circ x^{-1}\right)(C) \coloneqq \mu\left( \left\{ s \in \mathcal{S} : x(s) \in C \right\} \right)$ for each finite $C \subseteq \mathcal{C}$.
The interpretation is that the decision-maker
has a probabilistic \emph{belief} $\mu$ about the state
and cares only about the distribution of consequences,
so evaluates each act $x$ purely on the basis of its induced (subjective) lottery $\mu \circ x^{-1}$ over consequences.
The shape of $U$ captures \emph{risk attitude;}
subjective expected utility is the special case in which $U$ is linear.

An \emph{unambiguous equivalent} for $\mathord{\prefeq} \in \mathcal{P}^\circ$ of an act $x \in \mathcal{X}$ is an unambiguous act $x^\circ \in \mathcal{X}^\circ$ such that $x \prefeq x^\circ \prefeq x$.
A preference is called \emph{solvable} iff it has an unambiguous equivalent for every act.
When considering solvable preferences, it is natural (though not necessary) to assume that either the space $\mathcal{C}$ of consequences or the collection $\mathcal{E}^\circ$ of unambiguous events is `rich', since otherwise solvability is stringent.
Arbitrarily fix a map $e : \mathcal{P}^\circ \times \mathcal{X} \to \mathcal{X}^\circ$ such that $e(\mathord{\prefeq},x)$ is an unambiguous equivalent for $\prefeq$ of $x$ for each solvable preference $\mathord{\prefeq} \in \mathcal{P}^\circ$ and each act $x \in \mathcal{X}$,
with $e(\mathord{\prefeq},x) = x$ in case $x \in \mathcal{X}^\circ$.

Let $\mathcal{P}^{\mu,U}$
be the set of all solvable preferences $\mathord{\prefeq} \in \mathcal{P}^\circ$ such that $\mathord{\prefeq} |_{\mathcal{X}^\circ}$ has belief $\mu$ and risk attitude $U$.
Each preference $\mathord{\prefeq} \in \mathcal{P}^{\mu,U}$ may be viewed as a map $\mathcal{X} \to \R$, viz. the canonical utility representation $x \mapsto U\left( \mu \circ e(\mathord{\prefeq},x)^{-1} \right)$,
and so we may speak of \emph{pointwise compact} sets $P \subseteq \mathcal{P}^{\mu,U}$ of preferences:%
	\footnote{Recall that a set $\mathcal{F}$ of maps $\mathcal{A} \to \R$ is called pointwise compact exactly if for every $a \in \mathcal{A}$, $\left\{ f(a) : f \in \mathcal{F} \right\}$ is a compact subset of $\R$.}
explicitly, $P \subseteq \mathcal{P}^{\mu,U}$ is called pointwise compact iff for every act $x \in \mathcal{X}$, the set $\left\{ U\left( \mu \circ e(\mathord{\prefeq},x)^{-1} \right) : \mathord{\prefeq} \in P \right\}$ is a compact subset of $\R$.

\subsection{Maxmin preferences}
\label{sec:maxmin:def}

\begin{definition}
	\label{definition:maxmin}
	Given a belief $\mu : \mathcal{E}^\circ \to [0,1]$ and a risk attitude $U : \Delta(\mathcal{C}) \to \R$,
	a pointwise compact set $P \subseteq \mathcal{P}^{\mu,U}$
	is a \emph{maxmin representation} of a preference $\mathord{\prefeq^\star} \in \mathcal{P}$ iff
	$x \mapsto
	\min_{ \mathord{\prefeq} \in P } U\left( \mu \circ e(\mathord{\prefeq},x)^{-1} \right)$
	ordinally represents $\prefeq^\star$.
\end{definition}

Intuitively, such a decision-maker is unsure which preference $\mathord{\prefeq} \in P$ to evaluate acts according to, so cautiously values acts at their worst unambiguous equivalent among $\mathord{\prefeq} \in P$. An alternative interpretation is that there is a group of agents with preferences $P$, and that collective decisions are made according to the `Rawlsian' maxmin criterion.

In case no non-trivial event is unambiguous ($\mathcal{E}^\circ = \{\mathcal{S},\varnothing\}$),
unambiguous acts are precisely constants acts,
and thus unambiguous equivalents are \emph{certainty equivalents;}
then after we identify each consequence $c \in \mathcal{C}$ with the act $x_c \in \mathcal{X}^\circ$ constant at $c$ and with the lottery $\delta_c \in \Delta(\mathcal{C})$ degenerate at $c$,
a maxmin representation becomes $x \mapsto \min_{\mathord{\prefeq} \in P} U( e(\mathord{\prefeq},x) )$.
If in addition consequences are monetary prizes ($\mathcal{C} \subseteq \R$), then provided $U$ is strictly increasing, we may assume without loss that it is the identity $c \mapsto c$,
so that a maxmin representation is simply $x \mapsto \min_{\mathord{\prefeq} \in P} e(\mathord{\prefeq},x)$.

A special case of a maxmin representation is \emph{maxmin expected utility} \parencite{GilboaSchmeidler1989},
where $P$ consists entirely of preferences that are probabilistically sophisticated (across all acts $\mathcal{X}$, not just the unambiguous ones) with a common, linear risk attitude.
Precisely: each $\mathord{\prefeq} \in P$ is ordinally represented by
$x
\mapsto \int_{\mathcal{C}} u
\dd \left( \mu_{\prefeq} \circ x^{-1} \right)
= \int_{\mathcal{S}} ( u \circ x ) \dd \mu_{\prefeq}$
for some belief $\mu_{\prefeq} : \mathcal{E} \to [0,1]$,
where $u : \mathcal{C} \to \R$ is bounded
and $\left\{ \mu_{\mathord{\prefeq}} : \mathord{\prefeq} \in P \right\}$ is closed in the topology of setwise convergence (also known as `strong convergence').%
	\footnote{This implies that $P$ is pointwise compact:
	for any act $x \in \mathcal{X}$,
	the set $\left\{ \smallint_{\mathcal{S}} ( u \circ x ) \dd \mu_{\prefeq} : \mathord{\prefeq} \in P \right\}$
	is bounded since $u$ is,
	and is closed because
	$\smallint_{\mathcal{S}} (u \circ x) \dd \mu_n \to \smallint_{\mathcal{S}} (u \circ x) \dd \mu$
	whenever
	$\mu_n \to \mu$ setwise,
	since $u \circ x$ is bounded and measurable \parencite[see e.g. Theorem~2.1 in][]{JackaRoberts1997}.}

\subsection{Characterisation of maxmin preferences}
\label{sec:maxmin:charac}

\begin{proposition}[maxmin characterisation]
	\label{proposition:maxmin_join}
	Fix a belief $\mu : \mathcal{E}^\circ \to [0,1]$ and a risk attitude $U : \Delta(\mathcal{C}) \to \R$.
	For a pointwise compact set $P \subseteq \mathcal{P}^{\mu,U}$ and a preference $\mathord{\prefeq^\star} \in \mathcal{P}$, the following are equivalent:
	\begin{enumerate}

		\item \label{bullet:maxmin_maxmin}
		$P$ is a maxmin representation of $\prefeq^\star$.

		\item \label{bullet:maxmin_join}
		$\prefeq^\star$ is a minimum upper bound of $P$ with respect to `more uncertainty-averse than'.

	\end{enumerate}
\end{proposition}

\Cref{proposition:maxmin_join} furnishes an intuitive interpretation of maxmin preferences:
a preference with maxmin representation $P$ is precisely one that is more uncertainty-averse than any preference in $P$, but no more uncertainty-averse than that.

For the proof, let $\prefeq^\circ$ be the preference on $\mathcal{X}^\circ$ ordinally represented by $x \mapsto U\left( \mu \circ x^{-1} \right)$,
and define a binary relation $\primgeq$ on $\mathcal{X}$ by $x \primgeq y$ iff either (i)~$x=y$, (ii)~$x \in \mathcal{X}^\circ \not\ni y$ or (iii)~$x,y \in \mathcal{X}^\circ$ and $x \prefeq^\circ y$.
The central insight underlying the proof is that minimum upper bounds with respect to `more uncertainty-averse than' are precisely minimum upper bounds with respect to the single-crossing dominance relation $\SC$ induced by $\primgeq$.%
	\footnote{If $\mathord{\prefeq'}|_{\mathcal{X}^\circ} = \mathord{\prefeq^\circ}$, then
	$\mathord{\prefeq'}$ is more uncertainty-averse than $\mathord{\prefeq}$ iff $\mathord{\prefeq'} \SC \mathord{\prefeq}$.
	And clearly $\mathord{\prefeq'}|_{\mathcal{X}^\circ} = \mathord{\prefeq^\circ}$ for any upper bound $\prefeq'$ of $P$ with respect to `more uncertainty-averse than'.}

\begin{proof}
	Let $\prefeq^\circ$ and $\primgeq$ be as the previous paragraph.
	Fix a pointwise compact set $P \subseteq \mathcal{P}^{\mu,U}$ and a preference $\mathord{\prefeq^\star} \in \mathcal{P}$.

	Suppose that $P$ is a maxmin representation of $\prefeq^\star$.
	By the \hyperref[theorem:characterisation]{characterisation theorem}, it suffices to show that for all $x \primg y$ in $\mathcal{X}$,
	$x \prefeq^\star \mathrel{(\pref^\star)} y$
	iff there is a (strict) $P$-chain from $x$ to $y$.
	This holds when $x \notin \mathcal{X}^\circ$ since then $x,y$ are $\primgeq$-incomparable,
	and holds for $x,y \in \mathcal{X}^\circ$ since
	$\mathord{\primgeq}|_{\mathcal{X}^\circ} = \mathord{\prefeq^\star}|_{\mathcal{X}^\circ} = \mathord{\prefeq}|_{\mathcal{X}^\circ} = \mathord{\prefeq^\circ}$ for every $\mathord{\prefeq} \in P$.
	For $x \in \mathcal{X}^\circ \not\ni y$,
	there is a (strict) $P$-chain from $x$ to $y$
	iff $x \prefeq \mathrel{(\pref)} y$ for some $\mathord{\prefeq} \in P$
	iff $x \prefeq^\circ \mathrel{(\pref^\circ)} e(\mathord{\prefeq},y)$ for some $\mathord{\prefeq} \in P$
	iff $U\left( \mu \circ x^{-1} \right) \geq \mathrel{(>)} \min_{\mathord{\prefeq} \in P} U\left( \mu \circ e(\mathord{\prefeq},y)^{-1} \right)$
	iff $x \prefeq^\star \mathrel{(\pref^\star)} y$.

	For the converse, suppose that $\prefeq^\star$ is a minimum upper bound of $P$.

	\begin{namedthm}[Claim.]
		\label{claim:maxmin_join_solvable}
		$\prefeq^\star$ is solvable.
		(And hence belongs to $\mathcal{P}^{\mu,U}$.)
	\end{namedthm}

	\begin{proof}[Proof of the \protect{\hyperref[claim:maxmin_join_solvable]{claim}}]%
		\renewcommand{\qedsymbol}{$\square$}
		Fix an ambiguous act $x \in \mathcal{X} \setminus \mathcal{X}^\circ$;
		we seek an unambiguous equivalent $y \in \mathcal{X}^\circ$.
		Since $P$ is pointwise compact, there is a $\mathord{\prefeq'} \in P$ such that $y \coloneqq e(\mathord{\prefeq'},x)$ satisfies $U\left( \mu \circ y^{-1} \right) = \min_{\mathord{\prefeq} \in P} U\left( \mu \circ e(\mathord{\prefeq},x)^{-1} \right)$.
		Clearly $y \primgeq x$.
		We have $y \prefeq^\star x$ since $y \prefeq' x$ and $\prefeq^\star$ is an upper bound of $P \ni \mathord{\prefeq'}$.
		
		Choose $\mathord{\prefeq''} \in \mathcal{P}$ such that $P$ is a maxmin representation of $\prefeq''$.
		Then $x \prefeq'' y$ (and $y \prefeq'' x$) by definition of $\prefeq'$.
		As shown above, $\prefeq''$ must be an upper bound of $P$; thus $\mathord{\prefeq''} \SC \mathord{\prefeq^\star}$.
		It follows that $x \prefeq^\star y$.
	\end{proof}%
	\renewcommand{\qedsymbol}{$\blacksquare$}

	Now, since $\prefeq^\star$ is solvable and a minimum upper bound,
	$e(\mathord{\prefeq},\cdot) \prefeq^\circ e(\mathord{\prefeq^\star},\cdot)$ for every $\mathord{\prefeq} \in \mathcal{P}$, and
	$e(\mathord{\prefeq^\star},\cdot) \prefeq^\circ e(\mathord{\prefeq'},\cdot)$ for every $\mathord{\prefeq'} \in \mathcal{P}$ such that
	$e(\mathord{\prefeq},\cdot) \prefeq^\circ e(\mathord{\prefeq'},\cdot)$
	for every $\mathord{\prefeq} \in \mathcal{P}$.
	Equivalently, $x \mapsto U\left( \mu \circ e(\mathord{\prefeq^\star},x)^{-1} \right)$ is the pointwise greatest map $\mathcal{X} \to \R$ that is pointwise smaller than $x \mapsto U\left( \mu \circ e(\mathord{\prefeq},x)^{-1} \right)$ for every $\mathord{\prefeq} \in P$;
	in other words, it is
	$x \mapsto \min_{\mathord{\prefeq} \in P} U\left( \mu \circ e(\mathord{\prefeq},x)^{-1} \right)$.
\end{proof}

There is one subtlety in the (perhaps) intuitive equivalence asserted by \Cref{proposition:maxmin_join}: the role of pointwise compactness.
The equivalence generally fails when $P$ is not pointwise compact.%
	\footnote{When $P \subseteq \mathcal{P}^{\mu,U}$ is not pointwise compact, `$P$ is a maxmin representation of $\prefeq^\star \in \mathcal{P}$' is defined as in \Cref{definition:maxmin}, except with `$\min$' replaced by `$\inf$'.}
Firstly, $P$ can be a maxmin representation of a solvable preference, and yet admit no solvable minimum upper bound;%
	\footnote{Example: $\mathcal{E}^\circ = \{\mathcal{S},\varnothing\}$, $\mathcal{C} = [0,1]$, $U(x) = x$ for each $x \in \mathcal{C}$, and for some fixed $y \in \mathcal{X} \setminus \mathcal{C}$, $P \coloneqq \{ \mathord{\prefeq_\eps} : \eps \in (0,1] \}$, where $\prefeq_\eps$ belongs to $\mathcal{P}^{\mu,U}$ and satisfies $\eps \pref_\eps y \pref_\eps 0$ for each $\eps \in (0,1]$. Here $P$ is not pointwise closed.}
then \ref{bullet:maxmin_maxmin}$\implies$\ref{bullet:maxmin_join} fails,
and \ref{bullet:maxmin_join}$\implies$\ref{bullet:maxmin_maxmin} fails since (argued in the next paragraph) a minimum upper bound does exist.
Secondly, $P$ can admit several minimum upper bounds;%
	\footnote{Example: $\mathcal{E}^\circ = \{\mathcal{S},\varnothing\}$, $\mathcal{C} = \N$, $U(x) = -x$ for each $x \in \mathcal{C}$, and $P \coloneqq \{ \mathord{\prefeq_n} : n \in \N \}$, where each $\prefeq_n$ satisfies $e(\prefeq_n,x) = n$ for all $x \in \mathcal{X} \setminus \mathcal{C}$. Here $P$ is not pointwise bounded.}
then \ref{bullet:maxmin_join}$\implies$\ref{bullet:maxmin_maxmin} fails.

One consequence of \Cref{proposition:maxmin_join} is that every pointwise compact set $P \subseteq \mathcal{P}^{\mu,U}$ possesses exactly one minimum upper bound with respect to `more uncertainty-averse than'.
Existence can actually be proved directly, without using pointwise compactness: simply note that the relation $\primgeq$ is crown- and diamond-free, and apply the \hyperref[theorem:existence]{existence theorem}.%
	\footnote{A detail: $\primgeq$ need not be a partial order, as it may fail to be anti-symmetric on $\mathcal{X}^\circ$. But our theorems apply to the modified set of alternatives in which each $\prefeq^\circ$-equivalence class of unambiguous acts is treated as a single alternative.}
Uniqueness \emph{does} require pointwise compactness, by the previous paragraph. The reason why pointwise compactness yields uniqueness is that it implies (by the \hyperref[claim:maxmin_join_solvable]{claim} in the proof of \Cref{proposition:maxmin_join}) that all minimum upper bounds of $P$ are solvable. This suffices since by the \hyperref[theorem:characterisation]{characterisation theorem}, all minimum upper bounds of $P$ agree on $\primgeq$-comparable pairs $x,y \in \mathcal{X}$, and (by definition of $\primgeq$) this includes all pairs $x,y \in \mathcal{X}$ such that $x \in \mathcal{X}^\circ$.

\Cref{proposition:maxmin_join} also implies comparative statics for uncertainty-aversion:

\begin{corollary}[maxmin comparative statics]
	\label{corollary:maxmin_comp_stats}
	Fix a belief $\mu : \mathcal{E}^\circ \to [0,1]$ and a risk attitude $U : \Delta(\mathcal{C}) \to \R$,
	and let $\mathord{\prefeq},\mathord{\prefeq'} \in \mathcal{P}$ admit maxmin representations $P,P' \subseteq \mathcal{P}^{\mu,U}$.
	If $P'$ contains $P$, or if $P'$ dominates $P$ in the strong set order,%
		\footnote{The strong set order was defined in §\ref{sec:compstat:canonical}.}
	then $\prefeq'$ is more uncertainty-averse than $\prefeq$.
\end{corollary}

\subsection{Risk-aversion and caution}
\label{sec:maxmin:risk}

In this section, we translate to lotteries, characterising \emph{cautious} preferences as minimum upper bounds with respect to `more risk-averse than'.

There is a set $\mathcal{C}$ of consequences.
$\mathcal{X} \coloneqq \Delta(\mathcal{C})$ denotes all lotteries over consequences,
and $\mathcal{P}$ is the set of all preferences over $\mathcal{X}$.
We identify each degenerate lottery with the consequence that it delivers.

\begin{definition}[\cite{Yaari1969}]
	\label{definition:MRA}
	For two preferences $\mathord{\prefeq},\mathord{\prefeq'} \in \mathcal{P}$, we say that $\prefeq'$ is \emph{more risk-averse than} $\prefeq$ iff for any consequence $c \in \mathcal{C}$ and any lottery $x \in \mathcal{X}$, $c \prefeq \mathrel{(\pref)} x$ implies $c \prefeq' \mathrel{(\pref')} x$.
\end{definition}

Fix a function $e : \mathcal{P} \times \mathcal{X} \to \mathcal{C}$ such that $e(\mathord{\prefeq},x)$ is a certainty equivalent for $\prefeq$ of $x$ for each solvable $\mathord{\prefeq} \in \mathcal{P}$ and each $x \in \mathcal{X}$, with $e(\mathord{\prefeq},x) = x$ if $x \in \mathcal{X}^\circ$.
Given any $u : \mathcal{C} \to \R$, let $\mathcal{P}^u$ be the set of all solvable preferences $\mathord{\prefeq} \in \mathcal{P}$ whose restriction $\mathord{\prefeq}|_{\mathcal{C}}$ to consequences is ordinally represented by $u$.
Since each $\mathord{\prefeq} \in \mathcal{P}^u$ may be viewed as a map $\mathcal{X} \to \R$, namely its canonical representation $x \mapsto u( e(\mathord{\prefeq},x) )$, we may speak of pointwise compact sets $P \subseteq \mathcal{P}^u$.

\begin{definition}
	\label{definition:cautious}
	Given a utility function $u : \mathcal{C} \to \R$ over consequences,
	a pointwise compact set $P \subseteq \mathcal{P}^u$
	is a \emph{cautious representation} of a preference $\mathord{\prefeq^\star} \in \mathcal{P}$ iff
	$x \mapsto
	\min_{ \mathord{\prefeq} \in P } u( e(\mathord{\prefeq},x) )$
	ordinally represents $\prefeq^\star$.
\end{definition}

If consequences are monetary prizes ($\mathcal{C} \subseteq \R$),
then provided $u$ is strictly increasing,
we may assume without loss that it is the identity $c \mapsto c$,
in which case a cautious representation is simply $x \mapsto \min_{ \mathord{\prefeq} \in P } e(\mathord{\prefeq},x)$.

\emph{Cautious expected utility} \parencite{CerreiavioglioDillenbergerOrtoleva2015}
is the special case in which $P \subseteq \mathcal{P}^u$ contains only expected-utility preferences.

\begin{corollary}
	\label{corollary:caution_join}
	Fix a utility function $u : \mathcal{C} \to \R$.
	For a pointwise compact set $P \subseteq \mathcal{P}^u$ and a preference $\mathord{\prefeq^\star} \in \mathcal{P}$,
	$P$ is a cautious representation of $\prefeq^\star$
	iff
	$\prefeq^\star$ is a minimum upper bound of $P$ with respect to `more risk-averse than'.
\end{corollary}

\Cref{corollary:caution_join} implies comparative statics along the lines of \Cref{corollary:maxmin_comp_stats}, as well as the existence and uniqueness of minimum upper bounds of pointwise compact subsets of $\mathcal{P}^u$. Existence does not require pointwise compactness.



\begin{appendices}

\renewcommand*{\thesubsection}{\Alph{subsection}}

\crefalias{section}{appsec}
\crefalias{subsection}{appsec}
\crefalias{subsubsection}{appsec}

\section*{Appendices}
\label{app:appendix}
\addcontentsline{toc}{section}{Appendices}

\subsection{Standard definitions and notation}
\label{app:appendix:definitions}

This appendix collects definitions of standard order-theoretic concepts used in this paper.
Let $\mathcal{A}$ be a non-empty set, and $\abstrgeq$ a binary relation on it.

For $a,b \in \mathcal{A}$,
we write $a \not\abstrgeq b$ iff it is not the case that $a \abstrgeq b$,
and $a \abstrg b$ iff $a \abstrgeq b$ and $a \not\abstrgeq b$,
and we say that $a,b$ are \emph{$\abstrgeq$-(in)comparable} iff (n)either $a \abstrgeq b$ (n)or $b \abstrgeq a$.

$\abstrgeq$ is \emph{complete} iff every pair $a,b \in \mathcal{A}$ is $\abstrgeq$-comparable, \emph{transitive} iff $a \abstrgeq b \abstrgeq c$ implies $a \abstrgeq c$ for $a,b,c \in \mathcal{A}$, \emph{reflexive} iff $a \abstrgeq a$ for any $a \in \mathcal{A}$, and \emph{anti-symmetric} iff $a \abstrgeq b \abstrgeq a$ implies $a = b$ for $a,b \in \mathcal{A}$.
$\abstrgeq$ is a \emph{partial order} iff it is transitive, reflexive and anti-symmetric; $(\mathcal{A},\mathord{\abstrgeq})$ is then called a \emph{poset} (partially ordered set).
If $\abstrgeq$ is both complete and a partial order, then it is called a \emph{total order.}

Fix a subset $A \subseteq \mathcal{A}$.
$A$ is called \emph{totally ordered} iff all pairs $a,b \in A$ of distinct elements are $\abstrgeq$-comparable.
An element $a \in A$ is a \emph{minimum (maximum)} of $A$ iff $b \abstrgeq a$ ($a \abstrgeq b$) for every $b \in A$.

An element $b \in \mathcal{A}$ is an \emph{upper bound} of a set $A \subseteq \mathcal{A}$ iff $b \abstrgeq a$ for each $a \in A$,
and a \emph{minimum upper bound} iff in addition it is a minimum of the set of upper bounds of $A$.
Note that if $\abstrgeq$ is anti-symmetric, then the minimum upper bound of a set is unique if it exists.
Minimum upper bounds are also known as `least upper bounds', `suprema' or `joins'.
(Maximum) lower bounds are defined analogously; they are also known as `greatest lower bounds', `infima' or `meets'.

A partially ordered set $(\mathcal{A},\mathord{\abstrgeq})$ is a \emph{complete lattice} iff every subset of $\mathcal{A}$ has a minimum upper bound and a maximum lower bound, and simply a \emph{lattice} iff this is true for every two-element subset.

\subsection{Extension theorems}
\label{app:appendix:suzumura}

Extension theorems will be key to proving
\Cref{lemma:other_UB} (\cref{app:appendix:lemma_other_UB_pf}),
the \hyperref[theorem:existence]{existence theorem} (\cref{app:appendix:existence_pf}), and
\Cref{lemma:2_UBs} (\cref{app:appendix:2_UBs_pf}).

\begin{definition}
	\label{definition:extension}
	Let $\abstrgeq$ and $\abstrgeq'$ be binary relations on a set $\mathcal{A}$. $\abstrgeq'$ is an \emph{extension} of $\abstrgeq$ iff for $a,b \in \mathcal{A}$, $b \abstrgeq \mathrel{(\abstrg)} a$ implies $b \abstrgeq' \mathrel{(\abstrg')} a$.
\end{definition}

\begin{definition}
	\label{definition:suzumura}
	A binary relation $\abstrgeq$ on a set $\mathcal{A}$ is \emph{Suzumura consistent} iff for $a_1,\dots,a_K \in \mathcal{A}$, $a_1 \abstrgeq a_2 \abstrgeq \cdots \abstrgeq a_{K-1} \abstrgeq a_K$ implies that either $a_1 \abstrgeq a_K$ or $a_1,a_K$ are $\abstrgeq$-incomparable.
\end{definition}

\begin{namedthm}[Richter--Suzumura extension theorem.]
	\label{theorem:suzumura}
	A binary relation admits a complete and transitive extension iff it is Suzumura consistent.
\end{namedthm}

This result is due to \textcite{Suzumura1976}, and is a consequence of Richter's theorem (below). We use the \hyperref[theorem:suzumura]{Richter--Suzumura extension theorem} directly in the proof of the \hyperref[theorem:existence]{existence theorem} (\cref{app:appendix:existence_pf}).
In proving \Cref{lemma:2_UBs} (\cref{app:appendix:2_UBs_pf}), we rely on the following corollary.

\begin{namedthm}[Suzumura corollary.]
	\label{corollary:suzumura}
	Let $\abstrgeq$ be a transitive binary relation on a set $\mathcal{A}$,
	and let $a,b \in \mathcal{A}$ be such that $b \not\abstrgeq a$.
	Then $\abstrgeq$ admits a complete and transitive extension $\abstrgeq'$ such that $a \abstrg' b$.
\end{namedthm}

\begin{proof}
	Let $\abstrgeq^\triangle$ be the binary relation on $\mathcal{A}$ such that, for any $c,d \in \mathcal{A}$, $c \abstrgeq^\triangle d$ iff either (i)~$c \abstrgeq d$ or (ii)~$c = a$ and $d = b$.
	It suffices to show that $\abstrgeq^\triangle$ admits a complete and transitive extension to $\mathcal{A}$.
	So by the \hyperref[theorem:suzumura]{Richter--Suzumura extension theorem}, what we must show is that $\abstrgeq^\triangle$ is Suzumura consistent.

	To this end, let $a_1 \abstrgeq^\triangle \dots \abstrgeq^\triangle a_K$ in $\mathcal{A}$; we must establish that $a_K \ngtr^\triangle  a_1$.
	Let $\mathcal{I} = \{k \leq K : \text{$a_k = a$ and $a_{k+1} = b$} \}$, where $K+1$ is treated as $1$ by convention.
	If $\mathcal{I}$ is empty, then $a_1 \abstrgeq^\triangle a_K$ by transitivity of $\abstrgeq$.
	If $\mathcal{I}$ is a singleton $\mathcal{I} = \{k\}$,
	suppose toward a contradiction that $a_K \abstrg a_1$; then $b = a_{k+1} \abstrgeq a_k = a$ by transitivity of $\abstrgeq$, contradicting the hypothesis that $b \not\abstrgeq a$.
	Finally, suppose that $\abs*{\mathcal{I}} > 1$.
	Then there exist $k_1 < k_2$ such that $a_{k_1} = b$, $a_{k_2} = a$, and $a_k \abstrgeq a_{k+1}$ for all $k \in \{k_1,\dots,k_2-1\}$.
	It follows by transitivity of $\abstrgeq$ that $b \abstrgeq a$, a contradiction with $b \not\abstrgeq a$---thus $\abs*{\mathcal{I}} \leq 1$.
\end{proof}

The proof of \Cref{lemma:other_UB} (\cref{app:appendix:lemma_other_UB_pf}) uses a more general theorem.

\begin{definition}
	\label{definition:order_pair}
	An \emph{order pair} on a set $\mathcal{A}$ is a pair $(\mathord{\abstrgeq},\mathord{\abstrgeq'})$ where $\abstrgeq$ and $\abstrgeq'$ are binary relations on $\mathcal{A}$ and for any $a,b \in \mathcal{A}$, $a \abstrgeq' b$ implies $a \abstrgeq b$.
\end{definition}

\begin{definition}
	\label{definition:extension_pair}
	Let $(\mathord{\abstrgeq_1},\mathord{\abstrgeq_1'})$ and $(\mathord{\abstrgeq_2},\mathord{\abstrgeq_2})$ be order pairs on a set $\mathcal{A}$. $(\mathord{\abstrgeq_2},\mathord{\abstrgeq_2'})$ is an \emph{order-pair extension} of $(\mathord{\abstrgeq_1},\mathord{\abstrgeq_1'})$ iff for any $a,b \in \mathcal{A}$, $a \abstrgeq_1 b$ implies $a \abstrgeq_2 b$ and $a \abstrgeq_1' b$ implies $a \abstrgeq_2' b$.
\end{definition}

\begin{definition}
	\label{definition:acyclicity}
	Let $(\mathord{\abstrgeq},\mathord{\abstrgeq'})$ be an order pair on a set $\mathcal{A}$. A \emph{cycle} is a sequence $(a_k)_{k=1}^K$ in $\mathcal{A}$ such that $a_1 \abstrgeq a_2 \abstrgeq \cdots \abstrgeq a_{K-1} \abstrgeq a_K \abstrgeq' a_1$.
\end{definition}

\begin{namedthm}[Richter's theorem.]
	\label{theorem:richter}
	An order pair $(\mathord{\abstrgeq},\mathord{\abstrgeq'})$ on a set $\mathcal{A}$ features no cycles iff it admits an order-pair extension $(\mathord{\prefeq},\mathord{\pref})$ where $\prefeq$ is a preference on $\mathcal{A}$.
\end{namedthm}

This result is due to \textcite{Richter1966}.
See \textcite[][pp. 7--8]{ChambersEchenique2016} for a proof of the `only if' direction (the `if' direction is trivial).

\subsection{Proof of \texorpdfstring{\Cref{lemma:other_UB}}{Lemma \ref{lemma:other_UB}} (§\ref{sec:theory:characterisation}, p. \pageref{lemma:other_UB})}
\label{app:appendix:lemma_other_UB_pf}

We prove the contra-positive, using Richter's theorem (\cref{app:appendix:suzumura}).%
	\footnote{We thank an anonymous referee for suggesting this short argument.}
Suppose that all upper bounds $\prefeq''$ of $P$ satisfy $x \prefeq'' \mathrel{(\pref'')} y$;
we shall exhibit a (strict) $P$-chain from $x$ to $y$.
To that end, let $\abstrgeq$ be the binary relation on $\mathcal{X}$ such that for all $z,w \in \mathcal{X}$, $z \abstrgeq w$ iff either (i)~$z \primg w$ and there is a $P$-chain from $z$ to $w$ or (ii)~$z = y$ and $w = x$.
Let $\abstrgeq'$ be the binary relation on $\mathcal{X}$ such that for all $z,w \in \mathcal{X}$, $z \abstrgeq' w$ iff either (i)~$z \primg w$ and there is a strict $P$-chain from $z$ to $w$ or (ii)~$z = y$ and $w = x$
($z \abstrgeq' w$ iff $z \primg w$ and there is a strict $P$-chain from $z$ to $w$).
Clearly $(\mathord{\abstrgeq},\mathord{\abstrgeq'})$ is an order pair.

Any $\mathord{\prefeq''} \in \mathcal{P}$ such that $(\mathord{\prefeq''},\mathord{\pref''})$ is an order-pair extension of $(\mathord{\abstrgeq},\mathord{\abstrgeq'})$ must be an upper bound of $P$ by \Cref{lemma:UB_characterisation} (\cpageref{lemma:UB_characterisation}), and must satisfy $y \pref'' \mathrel{(\prefeq'')} x$ by construction of $\abstrgeq'$ (of $\abstrgeq$).
By hypothesis, no $\prefeq''$ with these properties exists.
Hence by \hyperref[theorem:richter]{Richter's theorem}, the order pair $(\mathord{\abstrgeq},\mathord{\abstrgeq'})$ features cycles.
Let $(z_k)_{k=1}^K$ be a cycle of minimal length.
Since $(z_k)_{k=1}^K$ is a cycle, there must be a $k' \leq \mathrel{(<)} K$ at which $z_{k'} = y$ and $z_{k'+1} = x$ (where $z_{K+1} \coloneqq z_1$ by convention), since otherwise $z_1 \primg \dots \primg z_K \primg z_1$.
Since $(z_k)_{k=1}^K$ is of minimal length, the index $k'$ is unique, and thus by definition of $\abstrgeq$ there is a $P$-chain from $z_k$ to $z_{k+1}$ for every $k \neq k'$ in $\{1,\dots,K-1\}$ (and by definition of $\abstrgeq'$ there is a strict $P$-chain from $z_K$ to $z_1$).
Hence $(w_k)_{k=1}^K \coloneqq (z_{k+k'})_{k=1}^K$, where the indices are interpreted modulo $K$, is a (strict) $P$-chain from $x$ to $y$.
\qed

\subsection{Proof that \texorpdfstring{\ref{bullet:ex_red} implies \ref{bullet:ex_ub} in the \hyperref[theorem:existence]{existence theorem}}{(1) implies (3) in the existence theorem} (§\ref{sec:theory:existence:theorem}, p. \pageref{theorem:existence})}
\label{app:appendix:existence_pf}

For a given $P \subseteq \mathcal{P}$, let $\prefeq^\circ$ be the (in general, incomplete) binary relation that satisfies conditions \ref{bullet:join_weak}--\ref{bullet:join_strict} in the \hyperref[theorem:characterisation]{characterisation theorem} for $\primgeq$-comparable pairs of alternatives, and that does not rank $\primgeq$-incomparable pairs of alternatives.%
	\footnote{Explicitly: for any $x,y \in \mathcal{X}$,
	if $x,y$ are $\primgeq$-incomparable then $x \nprefeq^\circ y \nprefeq^\circ x$, and
	if $x,y$ are $\primgeq$-comparable, wlog $x \primgeq y$, then
	$x \prefeq^\circ y$ iff there is a $P$-chain from $x$ to $y$, and
	$y \prefeq^\circ x$ iff there is no strict $P$-chain from $x$ to $y$.}
For each $P$, $\prefeq^\circ$ clearly exists and is unique.

In light of the \hyperref[theorem:characterisation]{characterisation theorem}, property \ref{bullet:ex_ub} in the \hyperref[theorem:existence]{existence theorem} requires precisely that $\prefeq^\circ$ admit a complete and transitive extension (i.e. an extension that lives in $\mathcal{P}$) for any $P \subseteq \mathcal{P}$.%
	\footnote{The term `extension' was defined in \cref{app:appendix:suzumura}.}
Our task is therefore to show that whenever $\primgeq$ is crown- and diamond-free, $\prefeq^\circ$ admits a complete and transitive extension for every $P \subseteq \mathcal{P}$.
We will use the \hyperref[theorem:suzumura]{Richter--Suzumura extension theorem} (\cref{app:appendix:suzumura}).

We first state two lemmata, then use them to show that \ref{bullet:ex_red} implies \ref{bullet:ex_ub}.
The remainder of this appendix is devoted to proving the lemmata.

\begin{definition}
	\label{definition:weak_trans}
	A binary relation $\abstrgeq$ on a set $\mathcal{A}$ is \emph{weakly transitive} iff
	for $a,b,c \in \mathcal{A}$, if $a \abstrgeq b \abstrgeq c$ and $a,c$ are $\abstrgeq$-comparable, then $a \abstrgeq c$.
\end{definition}

\begin{lemma}[weak transitivity of $\prefeq^\circ$]
	\label{lemma:weak_transitivity}
	Suppose that $\primgeq$ is diamond-free. Then $\prefeq^\circ$ is weakly transitive for any $P \subseteq \mathcal{P}$.
\end{lemma}

Weak transitivity is implied by Suzumura consistency (take $K=3$).
The converse is false in general,%
	\footnote{Consider $\mathcal{A} = \{a,b,c,d\}$ and the binary relation $\abstrgeq$ such that $a \abstrg b \abstrg c \abstrg d \abstrg a$ and no other pairs are $\abstrgeq$-comparable. This relation is weakly transitive, but clearly not Suzumura consistent.}
but true for $\prefeq^\circ$ when $\primgeq$ has no crowns:

\begin{lemma}[Suzumura consistency of $\prefeq^\circ$]
	\label{lemma:weak_trans_suzumura}
	Suppose that $\primgeq$ is crown-free. Then if $\prefeq^\circ$ is weakly transitive, it is Suzumura consistent.
\end{lemma}

\begin{proof}[Proof that \ref{bullet:ex_red} implies \ref{bullet:ex_ub}]
	Suppose that $\primgeq$ is crown- and diamond-free, and fix any $P \subseteq \mathcal{P}$.
	Since $\primgeq$ is diamond-free, $\prefeq^\circ$ is weakly transitive by \Cref{lemma:weak_transitivity}.
	Since $\primgeq$ is crown-free, it follows by \Cref{lemma:weak_trans_suzumura} that $\prefeq^\circ$ is Suzumura consistent.
	Invoking the \hyperref[theorem:suzumura]{Richter--Suzumura extension theorem} (\cref{app:appendix:suzumura}), we conclude that $\prefeq^\circ$ admits a complete and transitive extension.
\end{proof}

It remains to prove \Cref{lemma:weak_transitivity,lemma:weak_trans_suzumura}.
Begin with the former.
The role of diamond-freeness is to ensure a `crossing' property of decreasing sequences:

\begin{observation}
	\label{observation:diamonds_crossing}
	Suppose that $\primgeq$ is diamond-free, and consider $x,y,z \in \mathcal{X}$ with $x \primg z \primg y$.
	Let $(w_k)_{k=1}^K$ be a finite $\primgeq$-decreasing sequence with $w_1=x$ and $w_K=y$,
	and let $k'$ be the last $k<K$ such that $w_{k'} \primgeq z$.
	Then $z \primg w_{k'+1}$.
\end{observation}

\begin{proof}
	Since $w_1=x \primg z$, there exist $k<K$ such that $w_k \primgeq z$, so $k'$ is well-defined.
	It cannot be that $w_{k'+1} \primgeq z$, by definition of $k'$.
	Nor can it be that $z,w_{k'+1}$ are $\primgeq$-incomparable, for then $(x,z,w_{k'+1},y)$ is a diamond: $x \primgeq z \primgeq y$ (by hypothesis), $x \primgeq w_{k'+1} \primgeq y$ (since $(w_k)_{k=1}^K$ is $\primgeq$-decreasing from $x$ to $y$), and $z,w_{k'+1}$ are $\primgeq$-incomparable.
	Hence $z \primg w_{k'+1}$ by anti-symmetry of $\primgeq$.
\end{proof}

We will use the following piece of notation: given a set $\mathcal{A}$, a sequence $( a_n )_{n=1}^N$ in $\mathcal{A}$ and an element $b \in \mathcal{A}$, we write $(b) \union (a_n)_{n=1}^N$ for the sequence $(b,a_1,\dots,a_N)$, and similarly $(a_n)_{n=1}^N \union (b)$ for the sequence $(a_1,\dots,a_N,b)$.

\begin{proof}[Proof of \Cref{lemma:weak_transitivity}]
	Suppose that $\primgeq$ features no diamonds, and fix a $P \subseteq \mathcal{P}$ and $x,y,z \in \mathcal{X}$ such that $x \prefeq^\circ y \prefeq^\circ z$ and $x,z$ are $\prefeq^\circ$-comparable. We must show that $x \prefeq^\circ z$.
	This is immediate if $x,y,z$ are not distinct, so suppose that they are.
	Then by anti-symmetry of $\primgeq$, there are six cases to check: one for each strict ordering by $\primgeq$ of $x,y,z$.

	\emph{Case 1: $x \primg y \primg z$.} Since $x \prefeq^\circ y$ and $x \primg y$, by \ref{bullet:join_weak} there is a $P$-chain from $x$ to $y$. Similarly there is a $P$-chain from $y$ to $z$. The concatenation of these two $P$-chains is a $P$-chain from $x$ to $z$; hence $x \prefeq^\circ z$ by \ref{bullet:join_weak}.

	\emph{Case 2: $x \primg z \primg y$.} Since $x \primgeq y$ and $x \prefeq^\circ y$, there is a $P$-chain $(w_k)_{k=1}^K$ from $x$ to $y$ by \ref{bullet:join_weak}.
	Let $k'$ be the last $k < K$ for which $w_k \primgeq z$, so that $w_{k'} \primgeq z \primg w_{k'+1}$ by \Cref{observation:diamonds_crossing}. By definition of a $P$-chain, there is a preference $\prefeq$ in $P$ such that $w_{k'} \prefeq w_{k'+1}$.
	It must be that $w_{k'+1} \prefeq z$, for otherwise $( z ) \union ( w_k )_{k=k'+1}^K$ would be a strict $P$-chain from $z$ to $y$, in which case $y \nprefeq^\circ z$ by \ref{bullet:join_strict}, a contradiction.
	So we have $w_{k'} \prefeq w_{k'+1} \prefeq z$, which by transitivity of $\prefeq$ yields $w_{k'} \prefeq z$. It follows that $( w_k )_{k=1}^{k'} \union ( z )$ is a $P$-chain from $x$ to $z$, so that $x \prefeq^\circ z$ by \ref{bullet:join_weak}.

	\emph{Case 3: $y \primg x \primg z$.} This case is similar to the second. Since $y \primgeq z$ and $y \prefeq^\circ z$, there is a $P$-chain $( w_k )_{k=1}^K$ from $y$ to $z$ by \ref{bullet:join_weak}.
	Let $k'$ be the last $k < K$ for which $w_k \primgeq x$, so that $w_{k'} \primgeq x \primg w_{k'+1}$ by \Cref{observation:diamonds_crossing}. By definition of a $P$-chain, there is a preference $\prefeq$ in $P$ such that $w_{k'} \prefeq w_{k'+1}$.
	It must be that $x \prefeq w_{k'}$, for otherwise $( w_k )_{k=1}^{k'} \union ( x )$ would be a strict $P$-chain from $y$ to $x$, in which case $x \nprefeq^\circ y$ by \ref{bullet:join_strict}, a contradiction.
	So we have $x \prefeq w_{k'} \prefeq w_{k'+1}$, which by transitivity of $\prefeq$ yields $x \prefeq w_{k'+1}$. It follows that $( x ) \union ( w_k )_{k=k'+1}^K$ is a $P$-chain from $x$ to $z$, so that $x \prefeq^\circ z$ by \ref{bullet:join_weak}.

	\emph{Case 4: $y \primg z \primg x$.} Suppose toward a contradiction that $x \nprefeq^\circ z$. Then since $z \primgeq x$, by \ref{bullet:join_strict} there exists a strict $P$-chain from $z$ to $x$.
	Since $y \primgeq z$ and $y \prefeq^\circ z$, there is a $P$-chain from $y$ to $z$ by \ref{bullet:join_weak}.
	Concatenating these two $P$-chains yields a strict $P$-chain from $y$ to $x$, so that $x \nprefeq^\circ y$ by \ref{bullet:join_strict}, a contradiction.

	\emph{Case 5: $z \primg x \primg y$.} This case is similar to the fourth.
	Suppose toward a contradiction that $x \nprefeq^\circ z$. Then since $z \primgeq x$, by \ref{bullet:join_strict} there exists a strict $P$-chain from $z$ to $x$.
	Since $x \primgeq y$ and $x \prefeq^\circ y$, there is a $P$-chain from $x$ to $y$ by \ref{bullet:join_weak}.
	Concatenating these two $P$-chains yields a strict $P$-chain from $z$ to $y$, so that $y \nprefeq^\circ z$ by \ref{bullet:join_strict}, a contradiction.

	\emph{Case 6: $z \primg y \primg x$.} Suppose toward a contradiction that $x \nprefeq^\circ z$. Then by \ref{bullet:join_strict}, there is a strict $P$-chain $( w_k )_{k=1}^K$ from $z$ to $x$.
	Let $k'$ be the last $k < K$ for which $w_k \primgeq y$, so that $w_{k'} \primgeq y \primg w_{k'+1}$ by \Cref{observation:diamonds_crossing}. By definition of a $P$-chain, there is a preference $\prefeq$ in $P$ such that $w_{k'} \prefeq w_{k'+1}$. Since $( w_k )_{k=1}^K$ is a strict $P$-chain, there are $k''$ and $\mathord{\prefeq''} \in P$ be such that $w_{k''} \pref'' w_{k''+1}$, with $\mathord{\pref''} = \mathord{\pref}$ if $k'' = k'$.

	Case 6, sub-case A: $k'' < k'$. It must be that $y \pref w_{k'}$, for otherwise $( w_k )_{k=1}^{k'} \union ( y )$ would be a strict $P$-chain from $z$ to $y$, in which case $y \nprefeq^\circ z$ by \ref{bullet:join_strict}, a contradiction. So we have $y \pref w_{k'} \prefeq w_{k'+1}$, which by transitivity of $\prefeq$ yields $y \pref w_{k'+1}$. It follows that $( y ) \union ( w_k )_{k=k'+1}^K$ is a strict $P$-chain from $y$ to $x$, so that $x \nprefeq^\circ y$ by \ref{bullet:join_strict}, a contradiction.

	Case 6, sub-case B: $k'' \geq k'$. It must be that $y \prefeq w_{k'}$, for otherwise $( w_k )_{k=1}^{k'} \union ( y )$ would be a strict $P$-chain from $z$ to $y$, in which case $y \nprefeq^\circ z$ by \ref{bullet:join_strict}, a contradiction. So we have $y \prefeq w_{k'} \prefeq w_{k'+1}$, which by transitivity of $\prefeq$ yields $y \prefeq w_{k'+1}$, and $y \pref w_{k'+1}$ if $k'' = k'$. It follows that $( y ) \union ( w_k )_{k=k'+1}^K$ is a strict $P$-chain from $y$ to $x$, so that $x \nprefeq^\circ y$ by \ref{bullet:join_strict}, a contradiction.
\end{proof}

It remains to prove \Cref{lemma:weak_trans_suzumura}.

\begin{definition}
	\label{definition:weak_cycle}
	For a binary relation $\abstrgeq$ on a set $\mathcal{A}$, a \emph{weak cycle} is a finite sequence $(a_k)_{k=1}^K$ of distinct elements of $\mathcal{A}$ such that $a_k,a_{k+1}$ are $\abstrgeq$-comparable for each $k \in \{1,\dots,K\}$, where $a_{K+1}$ is understood as $a_1$.
\end{definition}

Note that crowns and diamonds are both examples of weak cycles.
The role of crown-freeness is to deliver a `shortcut' property of weak cycles:

\begin{observation}
	\label{observation:crowns_oldreducibility}
	Let $\abstrgeq$ be a transitive and crown-free binary relation on a set $\mathcal{A}$.
	Then for any weak cycle $(a_1,\dots,a_K)$ in $\abstrgeq$,
	$(a_1,a_k,a_K)$ is a weak cycle in $\abstrgeq$ for some $k \in \{2,\dots,K-1\}$.
\end{observation}

\begin{proof}
	We prove the contra-positive. Suppose that $\abstrgeq$ is transitive and that there is a weak cycle $(a_1,\dots,a_K)$ such that $(a_1,a_k,a_K)$ is not a weak cycle for any $k \in \{2,\dots,K-1\}$.
	In particular, choose $(a_1,\dots,a_K)$ to be a weak cycle of \emph{minimal length} with this property. Clearly its length $K$ is $\geq 4$.
	We will show that $(a_1,\dots,a_K)$ is a crown.

	There are two cases, $a_1 \abstrg a_2$ and $a_1 \abstrl a_2$. We consider the former case only; the latter is analogous.
	It must be that $a_2 \abstrl a_3$, for if $a_2 \abstrg a_3$ then $a_1,a_3$ are $\abstrgeq$-comparable by transitivity of $\abstrgeq$, so $(a_1,a_3,\dots,a_K)$ is a weak cycle for which $(a_1,a_k,a_K)$ is not a weak cycle for any $k \in \{3,\dots,K-1\}$, contradicting the minimality of $(a_1,\dots,a_K)$.
	Proceeding using the same argument, we conclude that $a_{k-1} \abstrg a_k \abstrl a_{k+1}$ for every even $k \in \{2,\dots,K-1\}$, where by convention $a_{K+1} \coloneqq a_1$.
	(This is apparent, separately, for $k < K$ even and for $k=K$.)

	Furthermore, it must be that for non-adjacent $k'<k''$ (i.e. those with $2 \leq k''-k' \leq K-2$), $a_{k'},a_{k''}$ are $\abstrgeq$-incomparable, for if they were comparable then $(a_1,\dots,a_{k'},a_{k''},\dots,a_K)$ would be a weak cycle for which $(a_1,a_k,a_K)$ is not a weak cycle for any $k \notin \{1,K\}$, contradicting the minimality of $(a_1,\dots,a_K)$.
	It follows that $(a_1,\dots,a_K)$ is a crown.
\end{proof}

\begin{proof}[Proof of \Cref{lemma:weak_trans_suzumura}]
	Suppose that $\primgeq$ is crown-free and that $\prefeq^\circ$ is weakly transitive.
	We must show that for every $K \geq 3$, the following claim holds:
	\begin{equation}
		\text{for any $x_1,\dots,x_K \in \mathcal{X}$,
		$x_1 \prefeq^\circ \cdots \prefeq^\circ x_K \prefeq^\circ x_1$
		implies $x_1 \prefeq^\circ x_K$.}
		\label{eq:CK}
		\tag*{$C(K)$}
	\end{equation}
	We proceed by strong induction on $K$.
	In the base case $K=3$, $C(3)$ is immediate from weak transitivity of $\prefeq^\circ$.

	Now fix any $K \geq 4$, and suppose that $C(K')$ holds for all $K' \leq K-1$; we will establish $C(K)$.
	Take any $x_1,\dots,x_K \in \mathcal{X}$ with $x_1 \prefeq^\circ \cdots \prefeq^\circ x_K \prefeq^\circ x_1$, wlog distinct.
	Since $\prefeq^\circ$ compares all and only $\primgeq$-comparable pairs of alternatives, $(x_1,\dots,x_K)$ is a weak cycle in $\primgeq$.
	Since $\primgeq$ (transitive and) crown-free, \Cref{observation:crowns_oldreducibility} implies the existence of a $k \in \{2,\dots,K-1\}$ such that $(x_1,x_k,x_K)$ is a weak cycle in $\primgeq$.
	We consider three cases.

	\emph{Case 1: $k=2$.}
	Since $x_2 \prefeq^\circ \dots \prefeq^\circ x_K \prefeq^\circ x_2$, the inductive hypothesis (in particular, $C(K-1)$) implies that $x_2 \prefeq^\circ x_K$. Since $x_1 \prefeq^\circ x_2$, and $x_1,x_K$ are $\prefeq^\circ$-comparable because they are $\primgeq$-comparable, it follows by weak transitivity that $x_1 \prefeq^\circ x_K$.

	\emph{Case 2: $k=K-1$.}
	This case is analogous to the first.

	\emph{Case 3: $2 < k < K-1$.}
	Since $x_1 \prefeq^\circ \cdots \prefeq^\circ x_k$ and $x_1,x_k$ are $\prefeq^\circ$-comparable (because they are $\primgeq$-comparable), the inductive hypothesis (in particular, $C(k)$) implies that $x_1 \prefeq^\circ x_k$.
	Similarly, $x_k \prefeq^\circ \cdots \prefeq^\circ x_K$, the $\primgeq$-comparability of $x_k,x_K$ and the inductive hypothesis yield $x_k \prefeq^\circ x_K$.
	Since $x_1,x_K$ are $\prefeq^\circ$-comparable (because they are $\primgeq$-comparable), it follows by weak transitivity that $x_1 \prefeq^\circ x_K$.
\end{proof}

\subsection{Proof of \texorpdfstring{\Cref{lemma:2_UBs}}{Lemma \ref{lemma:2_UBs}} (§\ref{sec:theory:uniqueness}, p. \pageref{lemma:2_UBs})}
\label{app:appendix:2_UBs_pf}

Since $\primgeq$ is transitive, and $x,y$ are $\primgeq$-incomparable, the \hyperref[corollary:suzumura]{Suzumura corollary} (\cref{app:appendix:suzumura}) implies that $\primgeq$ admits complete and transitive extensions $\mathord{\prefeq'},\mathord{\prefeq''} \in \mathcal{P}$ such that $x \pref' y$ and $y \pref'' x$.
$\prefeq'$ is an upper bound of $\mathcal{P}$ (and hence of any $P \subseteq \mathcal{P}$) because whenever $z \primgeq \mathrel{(\primg)} w$, we have $z \prefeq' \mathrel{(\pref')} w$ since $\prefeq'$ extends $\primgeq$; thus the consequent in the definition of `$\mathord{\prefeq'} \SC \mathord{\prefeq}$' (\cpageref{definition:SC}) is satisfied for any $z,w \in \mathcal{X}$ and $\mathord{\prefeq} \in \mathcal{P}$.
Similarly for $\prefeq''$.

\subsection{Results for maximum lower bounds}
\label{app:appendix:meets}

The \emph{inverse} of a binary relation $\abstrgeq$ on a set $\mathcal{A}$
is the binary relation $\abstrleq$ such that for $a,b \in \mathcal{A}$, $a \abstrleq b$ iff $b \abstrgeq a$.

\begin{observation}
	\label{observation:duality}
	If $\SC^{\mathord{\primgeq}}$ ($\SC^{\mathord{\primleq}}$) is the single-crossing-dominance relation induced by (the inverse of) the primitive order $\primgeq$ on $\mathcal{X}$, 
	then $\SC^{\mathord{\primleq}}$ is the inverse of $\SC^{\mathord{\primgeq}}$.
\end{observation}

Hence a maximum lower bound with respect to $\SC^{\mathord{\primgeq}}$ is precisely a minimum upper bound with respect to $\SC^{\mathord{\primleq}}$.
Since $\primleq$ is crown- and diamond-free (complete) iff $\primgeq$ is, the \hyperref[theorem:existence]{existence theorem} (\hyperref[proposition:uniqueness]{uniqueness proposition}) delivers:
\begin{corollary}[existence]
	\label{corollary:P_pre-lattice_join_meet}
	The following are equivalent:
	
	\begin{enumerate}[itemsep=0em,topsep=0em]

		\item Every set of preferences has a minimum upper bound.

		\item Every set of preferences has a maximum lower bound.

		\item $\primgeq$ is crown- and diamond-free.

	\end{enumerate}
\end{corollary}

\begin{corollary}[uniqueness]
	\label{corollary:uniqueness_join_meet}
	The following are equivalent:
	
	\begin{enumerate}[itemsep=0em,topsep=0em]

		\item Every set of preferences has \emph{at most} one minimum upper bound.

		\item Every set of preferences has \emph{exactly} one minimum upper bound.

		\item Every set of preferences has \emph{at most} one maximum lower bound.

		\item Every set of preferences has \emph{exactly} one maximum lower bound.

		\item $\primgeq$ is complete.

	\end{enumerate}
\end{corollary}

Finally, the analogue of the \hyperref[theorem:characterisation]{characterisation theorem} is as follows.

\begin{definition}
	\label{definition:revPchains}
	For a set $P \subseteq \mathcal{P}$ of preferences and two alternatives $y \primgeq x$ in $\mathcal{X}$, a \emph{reverse $P$-chain} from $y$ to $x$ is a finite sequence $(w_k)_{k=1}^K$ in $\mathcal{X}$ such that
	(i)~$w_1 = y$ and $w_K = x$,
	(ii)~for every $k<K$, $w_k \primleq w_{k+1}$, and
	(iii)~for every $k<K$, $w_k \prefeq w_{k+1}$ for some $\mathord{\prefeq} \in P$.
	A reverse $P$-chain is \emph{strict} iff $w_k \pref w_{k+1}$ for some $k<K$ and $\mathord{\prefeq} \in P$.
\end{definition}

\begin{corollary}[characterisation]
	\label{corollary:meet_charac}
	For a preference $\mathord{\prefeq_\star} \in \mathcal{P}$ and a set $P \subseteq \mathcal{P}$ of preferences, the following are equivalent:
	
	\begin{enumerate}[itemsep=0em,topsep=0em]

		\item $\prefeq_\star$ is a maximum lower bound of $P$.

		\item $\prefeq_\star$ satisfies: for any $x,y \in \mathcal{X}$ such that $y \primgeq x$,
		\begin{itemize}[itemsep=0em,topsep=0em]

			\item
			$x \prefeq^\star y$ iff there is a reverse $P$-chain from $x$ to $y$, and

			\item
			$y \prefeq^\star x$ iff there is no strict reverse $P$-chain from $x$ to $y$.

		\end{itemize}

	\end{enumerate}
\end{corollary}

\subsection{Proof of \texorpdfstring{\Cref{corollary:consensus_selection}}{Corollary \ref{corollary:consensus_selection}} (§\ref{sec:compstat:consensus}, p. \pageref{corollary:consensus_selection})}
\label{app:supp_appendix:consensus_selection}

Equip the space $2^{\mathcal{P}} \setminus \{\varnothing\}$ of non-empty sets of preferences with the $\SC$-induced strong set order $\gtrsim$.
It is easily verified that $\gtrsim$ is a partial order.
The consensus $C$ is a correspondence $2^{\mathcal{P}} \setminus \{\varnothing\} \rightrightarrows \mathcal{X}$, and \Cref{proposition:compstat_consensus} says precisely that it is increasing.

Let $\mathcal{Q}$ be the set of non-empty $P \subseteq \mathcal{P}$ at which $C(P)$ is non-empty.
The restriction of $C$ to $\mathcal{Q}$ is a non-empty-valued increasing correspondence into $\mathcal{X} \subseteq \R$.
It follows by Theorem 2.7 in \textcite{Kukushkin2013} that it admits an increasing selection $c : \mathcal{Q} \to \mathcal{X}$.%
	\footnote{This step is non-trivial because $\mathcal{X}$ and $\mathcal{P}$ need not be finite, nor even countable.}

Let
$\mathcal{U}(P)
\coloneqq \left\{ P'' \in \mathcal{Q} : P'' \gtrsim P \right\}$
for each non-empty $P \subseteq \mathcal{P}$,
and define $d : 2^{\mathcal{P}} \setminus \{\varnothing\} \to \mathcal{X}$ by
$d(P) \coloneqq \inf\left\{ c(P'') : P'' \in \mathcal{U}(P) \right\}$ if $\mathcal{U}(P) \neq \varnothing$
and $d(P) \coloneqq \sup \mathcal{X}$ otherwise.
The map $d$ is well-defined, and really does map into $\mathcal{X}$, because $(\mathcal{X},\mathord{\geq})$ is a complete lattice by the compactness of $\mathcal{X}$ and the Frink--Birkhoff theorem.%
	\footnote{See e.g. \textcite[Theorem 2.3.1]{Topkis1998}.}
Define an SCF $\phi : \Union_{n \in \N} \mathcal{P}^n \to \mathcal{X}$ by $\phi(\pi) \coloneqq d(\supp \pi)$ for every preference profile $\pi \in \Union_{n \in \N} \mathcal{P}^n$.

Since $d = c$ on $\mathcal{Q}$ (because $c$ is increasing) and $c$ is a selection from $C$, we have that $d(P) \in C(P)$ whenever the latter is non-empty, which is to say that the SCF $\phi$ respects unanimity.

To see that $\phi$ is monotone, consider $P,P' \in \mathcal{P}$ with $P' \gtrsim P$; we must show that $d(P') \geq d(P)$.
Observe that $\mathcal{U}(P') \subseteq \mathcal{U}(P)$.
If $\mathcal{U}(P')$ is empty then $d(P') = \sup \mathcal{X} \geq d(P)$ since $d(P) \in \mathcal{X}$.
If $\mathcal{U}(P)$ is empty then so is $\mathcal{U}(P')$, putting us in the previous case.
If neither is empty then
\begin{equation*}
	d(P')
	= \inf\left\{ c(P'') : P'' \in \mathcal{U}(P') \right\}
	\geq \inf\left\{ c(P'') : P'' \in \mathcal{U}(P) \right\}
	= d(P) .
	\qed
\end{equation*}

\subsection{Proof of \texorpdfstring{\Cref{proposition:compstat_consensus_2}}{Proposition \ref{proposition:compstat_consensus_2}} (§\ref{sec:compstat:consensus}, p. \pageref{proposition:compstat_consensus_2})}
\label{app:supp_appendix:compstat_consensus_2}

Fix $P,P' \subseteq \mathcal{P}$
such that $P'$ dominates $P$ in the $\SC$-induced strong set order
and $C(P) = C(P')$.
Suppose toward a contradiction that there are
alternatives $x > x'$
such that $x$ belongs to $C_2(P) \setminus C(P)$
and $x'$ belongs to $C_2(P') \setminus C(P')$,
but $x' \notin C_2(P) \setminus C(P)$.
(The other case, in which
$x \notin C_2(P') \setminus C(P')$ rather than
$x' \notin C_2(P) \setminus C(P)$,
is symmetric.)

We have $x \notin C_2(P')$ because
$x \notin C(P) = C(P')$
and $x \notin C_2(P') \setminus C(P')$.
Since $x' \in C_2(P')$,
it follows that there is an alternative $y \neq x'$ and a preference $\mathord{\prefeq'} \in P'$ such that $x' \pref' x$ and $y \pref' x$.

\begin{namedthm}[First claim.]
	\label{claim:x_good_in_P}
	$x \prefeq z$ for every $\mathord{\prefeq} \in P$ and each $z \in \mathcal{X}$ such that $x' \neq z < x$.
\end{namedthm}

\begin{proof}[Proof of the \protect{\hyperref[claim:x_good_in_P]{first claim}}]%
	\renewcommand{\qedsymbol}{$\square$}
	Suppose toward a contradiction that $z \pref x$ and $x' \neq z < x$ for some $\mathord{\prefeq} \in P$.
	Since also $x' \pref' x$ and $x' < x$,
	the maximum lower bound $\prefeq_\star$ of $\{ \mathord{\prefeq}, \mathord{\prefeq'} \}$
	must satisfy both $z \pref_\star x$ and $x' \pref_\star x$,
	so that $x \notin X_2(\mathord{\prefeq_\star})$.
	The preference $\prefeq_\star$ belongs to $P$
	since $P'$ dominates $P$ in the $\SC$-induced strong set order,
	so $C_2(P) \subseteq X_2(\mathord{\prefeq_\star})$.
	But then $x \notin C_2(P)$, a contradiction.
\end{proof}%
\renewcommand{\qedsymbol}{$\blacksquare$}

\begin{namedthm}[Second claim.]
	\label{claim:xprime_ge_x_in_P}
	$x' \prefeq x$ for every $\mathord{\prefeq} \in P$.
\end{namedthm}

\begin{proof}[Proof of the \protect{\hyperref[claim:xprime_ge_x_in_P]{second claim}}]%
	\renewcommand{\qedsymbol}{$\square$}
	We first show that $y>x$.
	Suppose toward a contradiction that $y<x$.
	On the one hand, the facts that $y \pref' x$, that $\prefeq'$ belongs to $P'$, and that $P'$ dominates $P$ in the $\SC$-induced strong set order
	together imply that there must be a $\mathord{\prefeq} \in P$ such that $y \pref x$.
	On the other hand, since $x' \neq y < x$,
	the \hyperref[claim:x_good_in_P]{first claim} requires that $x \prefeq y$.
	Contradiction!

	Now, to prove the claim, suppose toward a contradiction that $x \pref x'$ for some $\mathord{\prefeq} \in P$.
	Since $y \pref' x$ and $y > x > x'$,
	it follows that the minimum upper bound $\prefeq^\star$ of $\{ \mathord{\prefeq}, \mathord{\prefeq'} \}$
	satisfies $y \pref^\star x \pref^\star x'$,
	so that $x' \notin X_2(\mathord{\prefeq^\star})$.
	The preference $\prefeq^\star$ belongs to $P'$
	since $P'$ dominates $P$ in the $\SC$-induced strong set order,
	so $C_2(P') \subseteq X_2(\mathord{\prefeq^\star})$.
	But then $x' \notin C_2(P')$, a contradiction.
\end{proof}%
\renewcommand{\qedsymbol}{$\blacksquare$}

Since $x' \notin C(P') = C(P)$ by hypothesis, there is an alternative $z \in \mathcal{X}$ and a preference $\mathord{\prefeq} \in P$ such that $z \pref x'$,
whence $z \pref x$ by the \hyperref[claim:xprime_ge_x_in_P]{second claim}.
Then $z > x$ by the \hyperref[claim:x_good_in_P]{first claim},
and $X(\mathord{\prefeq}) = \{z\}$ since $x \in C_2(P) \subseteq X_2(\mathord{\prefeq})$.

\begin{namedthm}[Third claim.]
	\label{claim:prefpp}
	There is a preference $\mathord{\prefeq''} \in P'$ such that $x' \pref'' x$ and $z \pref'' w$ for every alternative $w \in \mathcal{X}$ satisfying $x' \neq w < z$.
\end{namedthm}

\begin{proof}[Proof of the \protect{\hyperref[claim:prefpp]{third claim}}]%
	\renewcommand{\qedsymbol}{$\square$}
	Let $\mathcal{Z}$ denote the $\prefeq'$-best alternatives in $\mathcal{X} \setminus \{x'\}$,
	and let $z'$ be the $\geq$-smallest element of $\mathcal{Z}$.
	If $z' = z$,
	then since $x' \pref' x$,
	we may take $\mathord{\prefeq''} \coloneqq \mathord{\prefeq'}$.
	Assume for the remainder that $z' \neq z$.

	Note that $z' \pref' x$ since $x' \neq y \pref' x$. 
	This implies that $z' > x$,
	since $z' < x$
	together with the fact that $P'$ dominates $P$ in the $\SC$-induced strong set order
	would imply the existence of a preference $\mathord{\prefeq^\dag} \in P$ such that $z' \pref^\dag x$,
	which is impossible by the \hyperref[claim:x_good_in_P]{first claim}.

	Write $\prefeq^\star$ for the minimum upper bound of $\{ \mathord{\prefeq}, \mathord{\prefeq'} \}$;
	we will show that we may take $\mathord{\prefeq''} \coloneqq \mathord{\prefeq^\star}$.
	Since $P'$ dominates $P$ in the $\SC$-induced strong set order,
	$\prefeq^\star$ lies in $P'$.
	It cannot be that $z' > z$,
	because this together with $z' \prefeq' z$ (by $z \neq x'$), $z > x$ and $z \pref x$
	would imply $z' \prefeq^\star z \pref^\star x'$,
	in which case $x' \notin X_2(\mathord{\prefeq^\star}) \supseteq C_2(P')$, a contradiction.
	Hence $z > z' > x$,
	which since $z \pref z' \pref' x$ implies that
	$z \pref^\star z' \pref^\star x$.
	Then $x' \pref^\star x$ since $x' \in C_2(P') \subseteq X_2(\mathord{\prefeq^\star})$.
	Furthermore, $z \pref^\star w$ for any $w < z$ since $z \pref w$.
\end{proof}%
\renewcommand{\qedsymbol}{$\blacksquare$}

Choose $\mathord{\prefeq''} \in P'$ as per the \hyperref[claim:prefpp]{third claim}.
Let $\prefeq_\star$ be the maximum lower bound of $\{ \mathord{\prefeq}, \mathord{\prefeq''} \}$,
and note that it belongs to $P$ since $P'$ dominates $P$ in the $\SC$-induced strong set order.
We have $x' \pref_\star x$ since $x' < x$ and $x' \pref'' x$.
We furthermore have $z \pref_\star x$
by the \hyperref[theorem:characterisation]{characterisation theorem} (the version for maximum lower bounds, given in \cref{app:appendix:meets}),
since $X(\mathord{\prefeq}) = \{z\}$ and $z \pref'' w$ for every alternative $w \in \mathcal{X}$ satisfying $z > w \geq x$,
so that there is no reverse $\{\prefeq,\prefeq''\}$-chain in from $x$ to $z$.
Thus $x \notin X_2(\mathord{\prefeq_\star}) \supseteq C_2(P)$, a contradiction.
\qed

\subsection{Proof of \texorpdfstring{\hyperref[proposition:bottomtop]{\Cref*{proposition:bottomtop_weak}$^\star$}}{Proposition \ref{proposition:bottomtop_weak}*} (§\ref{sec:compstat:robust}, p. \pageref{proposition:bottomtop})}
\label{app:supp_appendix:proposition:bottomtop}

Let $\prefeq^\star$ denote the minimum upper bound of $P$.
We shall show that \ref{item:rich} is equivalent to \ref{item:tight_bound} and (separately) that \ref{item:tight_bound} is equivalent to \ref{item:consistent}.

\ref{item:rich} implies \ref{item:tight_bound} by \Cref{proposition:bottomtop_weak}.
For the converse, suppose that $P$ is not rich;
we will show that \ref{item:tight_bound} fails.
By hypothesis, there exists a non-empty menu $M = \{x_0,x_1,\dots,x_K\} \subseteq \mathcal{X}$
and preferences $\mathord{\prefeq_1},\dots,\mathord{\prefeq_K} \in P$
such that $x_0 \prefeq_1 x_1 \prefeq_2 x_2 \prefeq_3 \cdots \prefeq_K x_K$,
and yet $x_0 \notin X_M(\mathord{\prefeq})$ for every $\mathord{\prefeq} \in P$.
Let $\primgeq$ be a total order on $\mathcal{X}$ such that $x_0 \primg \dots \primg x_K$.
Then $(x_0,x_1,\dots,x_K)$ is a $P$-chain,
so $x_0 \in X_M(\mathord{\prefeq^\star})$ by the \hyperref[theorem:characterisation]{characterisation theorem}.
It follows that $\max X_M(\mathord{\prefeq^\star}) = x_0 > \max X_M(P)$, so \ref{item:tight_bound} fails.

\ref{item:tight_bound} immediately implies \ref{item:consistent}.
To prove the converse, we shall demonstrate that for any given total order $\primgeq$ on $\mathcal{X}$,
if there is a preference $\mathord{\prefeq} \in \mathcal{P}$ such that $\max X_M(P) = \max X_M(\mathord{\prefeq})$ for every non-empty menu $M \subseteq \mathcal{X}$,
then $\prefeq^\star$ is such a preference.
So fix a total order $\primgeq$ on $\mathcal{X}$,
and suppose toward a contradiction that $\mathord{\prefeq} \in \mathcal{P}$ has the requisite property
whereas $\prefeq^\star$ does not.
Then there is a non-empty menu $M \subseteq \mathcal{X}$
such that $x \coloneqq \max X_M(\mathord{\prefeq}) = \max X_M(P) \neq \max X_M(\mathord{\prefeq^\star}) \eqqcolon y$.
It must be that $x < y$, since (as argued in the proof of \Cref{proposition:bottomtop_weak}) we have $x \leq y$ by the \hyperref[theorem:MS]{MCS theorem} and the fact that $\prefeq^\star$ is an upper bound of $P$.
It follows by the \hyperref[theorem:characterisation]{characterisation theorem} that there is a $P$-chain $(w_k)_{k=1}^K$ from $y$ to $x$, so that
$\max X_{\{w_k,w_{k+1}\}}(P)
= \{w_k\}$
for every $k \in \{1,\dots,K-1\}$.
Then because (by hypothesis) $\max X_{M'}(\mathord{\prefeq}) = \max X_{M'}(P)$ for every non-empty menu $M' \subseteq \mathcal{X}$,
it must be that $y = w_1 \prefeq w_2 \prefeq \cdots \prefeq w_K = x$.
Since $x \in X_M(\mathord{\prefeq})$ by hypothesis and $y \in M$,
it follows that $y \in X_M(\mathord{\prefeq})$.
But then 
$\max X_M(\mathord{\prefeq}) \geq y > x = \max X_M(\mathord{\prefeq})$, which is absurd.
\qed

\end{appendices}



\printbibliography[heading=bibintoc]


\end{document}

%% file: preamble.tex
\DeclareTextFontCommand{\emph}{\slshape}

\makeatletter
\renewcommand{\paragraph}{%
	\@startsection{paragraph}{4}%
	{\z@}{1.75ex \@plus 1ex \@minus .2ex}{-0.7em}%
	{\normalfont\normalsize\bfseries}%
}
\makeatother

\usepackage[T1]{fontenc}
\usepackage{lmodern}
\usepackage[utf8]{inputenc}

\usepackage[american,german,british]{babel}

\usepackage[activate={true,nocompatibility}]{microtype}

\usepackage{csquotes}

\usepackage[backend=biber,autolang=other,style=apa,maxcitenames=5,sortcites=false]{biblatex}
\DeclareLanguageMapping{british}{british-apa}
\DeclareLanguageMapping{american}{american-apa}

\usepackage{amsmath}
\usepackage{amssymb}
\usepackage{amsfonts}
\usepackage{amsthm}
\usepackage{mathtools}
\usepackage{stmaryrd}

\let\originalleft\left
\let\originalright\right
\renewcommand{\left}{\mathopen{}\mathclose\bgroup\originalleft}
\renewcommand{\right}{\aftergroup\egroup\originalright}

\usepackage{graphicx}
\usepackage{tikz,pgfplots}
\pgfplotsset{compat=1.10}
\usetikzlibrary{shapes.multipart}
\usetikzlibrary{decorations.markings}

\usepackage[title]{appendix}

\usepackage{datetime}
\newdateformat{datestyle}{ {\THEDAY} \monthname[\THEMONTH] \THEYEAR }

\usepackage{enumerate}
\usepackage{enumitem}
\setlist[enumerate,1]{label=(\arabic*)}
\setlist[itemize,1]{label=--}
\setlist[itemize,2]{label=--}
\setlist[itemize,3]{label=--}
\setlist[itemize,4]{label=--}

\usepackage{xcolor}
\usepackage{verbatim}
\usepackage{gensymb}
\usepackage{rotating}
\usepackage{subcaption}
\usepackage{floatpag}
\floatpagestyle{plain}
\usepackage{marvosym}

\usepackage{hyphenat}
\theoremstyle{definition}

\newtheorem{proposition}{Proposition}
\newtheorem{lemma}{Lemma}
\newtheorem{corollary}{Corollary}
\newtheorem{remark}{Remark}
\newtheorem{observation}{Observation}
\newtheorem{example}{Example}
\newtheorem{definition}{Definition}

\newtheoremstyle{named}
	{\topsep}					
	{\topsep}					
	{}							
	{0pt}						
	{\bfseries}					
	{}							
	{5pt plus 1pt minus 1pt}	
	{\thmnote{#3}}				
\theoremstyle{named}
\newtheorem{namedthm}{}

\renewcommand{\qedsymbol}{$\blacksquare$}

\usepackage{xpatch}
\xpatchcmd{\proof}{\itshape}{\proofheadfont}{}{}
\newcommand{\proofheadfont}{\slshape}

\usepackage{hyperref}
\hypersetup{pdfborder=0 0 0}

\usepackage[nameinlink]{cleveref}
\crefname{page}{p.}{pp.}
\crefname{equation}{equation}{equations}
\crefname{section}{section}{sections}
\crefname{subsection}{section}{sections}
\crefname{subsubsection}{section}{sections}
\crefname{appsec}{appendix}{appendices}
\crefname{supplsec}{supplemental appendix}{supplemental appendices}
\crefname{footnote}{footnote}{footnotes}
\crefname{figure}{figure}{figures}
\crefname{table}{table}{tables}
\crefname{theorem}{theorem}{theorems}
\crefname{proposition}{proposition}{propositions}
\crefname{lemma}{lemma}{lemmata}
\crefname{corollary}{corollary}{corollaries}
\crefname{remark}{remark}{remarks}
\crefname{observation}{observation}{observations}
\crefname{example}{example}{examples}
\crefname{fact}{fact}{facts}
\crefname{definition}{definition}{definitions}
\crefname{assumption}{assumption}{assumptions}
\crefname{exercise}{exercise}{exercises}
\crefname{notation}{notation}{notation}
\crefname{claim}{claim}{claims}
\crefname{conjecture}{conjecture}{conjectures}

\newcommand{\eps}{\varepsilon}

\newcommand{\dd}{\mathrm{d}}

\DeclareMathOperator*{\supp}{supp}

\newcommand{\R}{\mathbf{R}}

\newcommand{\N}{\mathbf{N}}

\newcommand{\join}{\vee}

\newcommand{\meet}{\wedge}

\newcommand{\union}{\cup}

\newcommand{\Union}{\bigcup}

\newcommand{\Intersect}{\bigcap}

\DeclarePairedDelimiter\abs{\lvert}{\rvert}

\newcommand*{\xslant}[2][76]{%
	\begingroup
	\sbox0{#2}%
	\pgfmathsetlengthmacro\wdslant{\the\wd0 + cos(#1)*\the\wd0}%
	\leavevmode
	\hbox to \wdslant{\hss
		\tikz[
			baseline=(X.base),
			inner sep=0pt,
			transform canvas={xslant=cos(#1)},
		] \node (X) {\usebox0};%
		\hss
		\vrule width 0pt height\ht0 depth\dp0 %
	}%
	\endgroup
}
\makeatletter
\newcommand*{\xslantmath}{}
\def\xslantmath#1#{%
	\@xslantmath{#1}%
}
\newcommand*{\@xslantmath}[2]{%
	\ensuremath{%
		\mathpalette{\@@xslantmath{#1}}{#2}%
	}%
}
\newcommand*{\@@xslantmath}[3]{%
	\xslant#1{$#2#3\m@th$}%
}
\makeatother

\makeatletter
\def\namedlabel#1#2{\begingroup
	#2%
	\def\@currentlabel{#2}%
	\phantomsection\label{#1}\endgroup
}
\makeatother

\makeatletter
\let\save@mathaccent\mathaccent
\newcommand*\if@single[3]{%
	\setbox0\hbox{${\mathaccent"0362{#1}}^H$}%
	\setbox2\hbox{${\mathaccent"0362{\kern0pt#1}}^H$}%
	\ifdim\ht0=\ht2 #3\else #2\fi
	}
\newcommand*\rel@kern[1]{\kern#1\dimexpr\macc@kerna}
\newcommand*\widebar[1]{\@ifnextchar^{{\wide@bar{#1}{0}}}{\wide@bar{#1}{1}}}
\newcommand*\wide@bar[2]{\if@single{#1}{\wide@bar@{#1}{#2}{1}}{\wide@bar@{#1}{#2}{2}}}
\newcommand*\wide@bar@[3]{%
	\begingroup
	\def\mathaccent##1##2{%
	  \let\mathaccent\save@mathaccent
	  \if#32 \let\macc@nucleus\first@char \fi
	  \setbox\z@\hbox{$\macc@style{\macc@nucleus}_{}$}%
	  \setbox\tw@\hbox{$\macc@style{\macc@nucleus}{}_{}$}%
	  \dimen@\wd\tw@
	  \advance\dimen@-\wd\z@
	  \divide\dimen@ 3
	  \@tempdima\wd\tw@
	  \advance\@tempdima-\scriptspace
	  \divide\@tempdima 10
	  \advance\dimen@-\@tempdima
	  \ifdim\dimen@>\z@ \dimen@0pt\fi
	  \rel@kern{0.6}\kern-\dimen@
	  \if#31
	    \overline{\rel@kern{-0.6}\kern\dimen@\macc@nucleus\rel@kern{0.4}\kern\dimen@}%
	    \advance\dimen@0.4\dimexpr\macc@kerna
	    \let\final@kern#2%
	    \ifdim\dimen@<\z@ \let\final@kern1\fi
	    \if\final@kern1 \kern-\dimen@\fi
	  \else
	    \overline{\rel@kern{-0.6}\kern\dimen@#1}%
	  \fi
	}%
	\macc@depth\@ne
	\let\math@bgroup\@empty \let\math@egroup\macc@set@skewchar
	\mathsurround\z@ \frozen@everymath{\mathgroup\macc@group\relax}%
	\macc@set@skewchar\relax
	\let\mathaccentV\macc@nested@a
	\if#31
	  \macc@nested@a\relax111{#1}%
	\else
	  \def\gobble@till@marker##1\endmarker{}%
	  \futurelet\first@char\gobble@till@marker#1\endmarker
	  \ifcat\noexpand\first@char A\else
	    \def\first@char{}%
	  \fi
	  \macc@nested@a\relax111{\first@char}%
	\fi
	\endgroup
}
\makeatother
